\documentclass[a4paper,12pt]{report}
\newtheorem{theorem}{Theorem}

\newtheorem{proposition}[theorem]{Proposition}

\newenvironment{proof}[1][Proof]{\textbf{#1.}}{\ \rule{0.5em}{0.5em}}

\usepackage{vmargin}
\setmarginsrb{3cm}{2.75cm}{3cm}{2.75cm}{0mm}{0mm}{0mm}{10mm}
\usepackage{latexsym}
\usepackage{amsmath}
\usepackage{amssymb}
\usepackage{graphicx}
\usepackage{color}
\usepackage{hyperref}
\hypersetup{colorlinks=true,linkcolor=black,urlcolor=black,citecolor=black}
\usepackage[font=small]{caption}
\usepackage{cite}
\usepackage{url}
\usepackage{abstract}

\usepackage{etoolbox}
\patchcmd{\abstract}{\titlepage}{\cleardoublepage}{}{}
\patchcmd{\endabstract}{\endtitlepage}{\clearpage}{}{}
\usepackage{multicol}
\begin{document}
\title{\bf{Master Thesis} \vspace*{1cm}\\ 
	Wave Group Evolution and Interaction}
\author{\bf{Natanael Karjanto} \vspace*{1cm}\\ 
Supervisors:\\ 
\textbf{Prof. Dr. E. (Brenny) van Groesen}\\ 
\textbf{Ir. Gert Klopman}\\
\textbf{Dr. Andonowati} \\ 
\textit{Department of Applied Mathematics, University of Twente} \\
\textit{P.O. Box 217, 7500 AE Enschede, The Netherlands}}
\date{\small 13 June 2003}
\maketitle

\begin{abstract}
\setcounter{page}{2}
\addcontentsline{toc}{chapter}{Abstract}
This thesis deals with mathematical and physical aspects of deterministic freak wave generation in a hydrodynamic laboratory. We adopt the nonlinear Schr\"odinger (NLS) equation as a mathematical model for the evolution of the surface gravity wave packet envelopes. As predicted theoretically and observed experimentally by Benjamin and Feir (1967), Stokes waves on deep water are unstable under a sideband modulation. The plane-wave solution of the NLS equation, which represents the fundamental component of the Stokes wave, experiences a similar modulational instability under the linear perturbation theory. A nonlinear extension for this perturbed wave is given by the family of Soliton on Finite Background (SFB), where an exact expression is available and can be derived analytically by several techniques. We discuss the characteristics of the SFB family as a prime candidate for freak wave events. The corresponding physical wave field exhibits intriguing phenomena of wavefront dislocation and phase singularity.
\end{abstract}

\chapter*{Acknowledgement}
This work is executed at the University of Twente, The Netherlands as part of the project `Prediction and generation of deterministic extreme waves in hydrodynamic laboratories' (TWI.5374) of the Netherlands Organization of Scientific Research NWO, subdivision Applied Sciences STW.
\addcontentsline{toc}{chapter}{Acknowledgement}

The author is grateful for the superlative supervision, patient teaching, fruitful discussion, upbuilding advice, joyful friendship, and incessant prayer to the following individuals. (The list of names is ordered alphabetically according to surname.) 
\begin{center}
\begin{minipage}{13cm}
\begin{multicols}{2}
\noindent
Andonowati \\
Jurjen Battjes \\
Edi Cahyono \\
Bernard Geurts \\
Manfred Hammer \\
Kirankumar Hiremath \\
Ren\'e Huijsmans \\
Salemah Ismail \\
Gerard Jeurnink \\
Sandra Kamphuis \\
Zakaria Karjanto \\
Gert Klopman \\
Giok Lien Lauw \\
Helena Margaretha \\
Gjerrit Meinsma \\	
Jaap Molenaar \\
Abdoreza Moqadasi \\
Arthur Mynett \\
Jacqueline Nicolau \\
Sri Nurdiati \\
Miguel Onorato \\
Lars Pesch \\
Pearu Peterson \\
Marielle Plekenpol \\
Gerhard Post \\
Serevia Revita \\
Angela Sembiring \\
Edy Soewono \\
Ardhasena Sopahelukawan \\
Agus Suryanto \\
Hadi Susanto \\
Wooi Nee Tan \\
Ferli Tiani \\
Frits van Beckum \\
Imelda van de Voorde \\
Jaap van der Vegt \\
Christa van der Meer \\
Stephan van Gils \\
E. (Brenny) van Groesen \\
Willem Visser \\
Katarzyna Wac \\
Elizabeth Yasmine Wardoyo
\end{multicols} 
\end{minipage}
\end{center}

\tableofcontents

\chapter{Introduction} \label{introduction}

This Master thesis is the preparatory stage for the PhD project `Prediction and generation of deterministic extreme waves in hydrodynamic laboratories'. The term \textit{`extreme waves'} refers to very high amplitude and very steep waves, whose heights exceed a factor of 2.2 times the significant wave height of a measured wave train (see for example~\cite{Dean90}). Such waves are also described as \textit{`freak waves'}, \textit{`rogue waves'}, or \textit{`giant waves'}. According to~\cite{Osbo00}, these are rare events in nonlinear gravity waves which occasionally rise to surprising heights (as compared to linear theory of Gaussian random processes) above the background wave field. Freak waves occur in bad weather conditions when the average wave height is high and several big waves come together to create a `monster'. Detailed knowledge about their appearance is not yet available, but in the literature four main reasons are described: \\
(1) wave-current interaction: when wind pushes against a strong current (e.g. in South Africa); \\
(2) wave-bathymetry interaction: when a shallow sea bottom focuses waves to one spot (e.g. in Norway); \\
(3) extreme wave statistics: by chance (1 in 10,000-year statistic, according to the linear theory of Gaussian random processes); and  \\
(4) wave nonlinearity: when waves become unstable and start to self-focus. \\
(See the programme summary of `Freak Wave' in~\cite{freakwave,freakwaveqa}.)

In this thesis, we focus on mathematical modelling of the extreme waves and interpretation of the mathematical results into a physical situation at sea. One specific model for describing envelope surface waves is the Nonlinear Schr\"{o}dinger (NLS) equation. This equation (or higher-order extensions) can describe many of the features of the dynamics of extreme waves, which are thought to arise as a result of a \textit{modulational instability} phenomenon. This process appears in many nonlinear wave systems when a monochromatic (plane) wave is perturbed. Because of this instability, small amplitude perturbations grow rapidly under the combined effects of nonlinearity and dispersion~\cite{KipD00}. In Chapter~\ref{BF_SFB},\footnote[1]{This chapter can also be found in~\cite{Karj02}.} we consider one example of the modulational instability phenomenon. We analyze the evolution of the wave group envelope, using both linear and nonlinear theories. Linear theory predicts exponential growth of the amplitude when certain conditions are satisfied---the \textit{Benjamin--Feir (BF) instability}~\cite{Scot99}. However, when the amplitude becomes large, higher-order nonlinear effects must be taken into account, that, as it turns out, will prevent further exponential growth. To investigate this growth in detail, we study a nonlinear model of BF instability in the NLS regime, namely the \textit{Soliton on Finite Background} (SFB). It is found that in the range of instability, the maximum amplitude amplification factor can reach at most three times of the initial amplitude. A similar NLS solution obtained using the `inverse scattering transform' is also considered in~\cite{Osbo00}, with amplitude amplification factor 2.4.

In Chapter~\ref{WD_PS},\footnote[7]{This chapter is also part of~\cite{Groe02}.} we study the physical wave profile that describes the surface elevation of the waves. It is found that the physical wave profiles corresponding to SFB of the NLS equation have a \textit{wave dislocation} phenomenon for perturbation wavenumbers smaller than a certain critical value in the instability region. Eventually, the appearance of wave dislocations is related to the real amplitude having zeros at time $t = 0$. The corresponding trajectories in the dispersion plane of the local wavenumber and the local frequency show \textit{phase singularities}. Furthermore, we derive the \textit{phase--amplitude} equations for the NLS equation which are coupled relations between the phase and the amplitude. The phase equation is known as the \textit{nonlinear dispersion relation} and the amplitude equation is related to the \textit{wave action conservation equation}. It is observed that vanishing of the amplitude is related to the phase singularity and that the focusing effect between the wave dislocations leads to large amplitude amplifications.

In Chapter~\ref{perspective}, we give a motivation for studying wave group interaction; the relation among bi-harmonic, positively modulated bi-harmonic, and Benjamin--Feir modulated wave groups, as well as their properties. We also present an overview of some papers discussing the occurrence of extreme waves in the random oceanic sea. Several conclusions are given in Chapter~\ref{conclusions}, some open problems and future research are given in Chapter~\ref{open_problems}.

\chapter{Benjamin--Feir Instability and Soliton on Finite Background} \label{BF_SFB}

\section{Modelling of Wave Envelopes}

\subsection{Linear Theory}

In the linear theory of water waves, we can restrict the analysis of a surface elevation $\eta(x,t)$ to a \textit{one-mode solution} of the form $\eta(x,t) = a\,e^{i(k x - \omega t)} + \text{c.c.}$ (complex conjugate), where $a$ is constant amplitude, $k$ is wavenumber, and $\omega$ is frequency. Then a general solution is a superposition of one-mode solutions. For surface waves on water of constant depth $h$ measured from the still water level, the parameters $k$ and $\omega$ are related by the following \textit{linear dispersion relation}~\cite{Debn94}:
\begin{equation}
\omega^{2} = g\,k\,\tanh k\,h,
\end{equation}
where $g$ is the gravitational acceleration. We also write the linear dispersion relation as $\omega \equiv \Omega(k) = k\,\sqrt{\frac{g\,\tanh k\,h}{k}}$. This dispersion relation can be derived from the linearized equations of the full set of equations for water waves. The \textit{phase velocity} is defined as $\frac{\Omega(k)}{k}$ and the \textit{group velocity} is defined as $\frac{d\Omega}{dk} = \Omega'(k)$.

Since the linear wave system has elementary solutions of the form $e^{\,i\,(k\,x - \Omega(k)\,t)}$, it is often convenient to write the general solution of an initial value problem as an integral of its Fourier components~\cite{Scot99}:
\begin{equation}
\eta(x,t) := \int^{\,\infty}_{\!\!-\infty}\!\!\alpha(k)\:e^{\,i\,(k\,x - \Omega(k)\,t)} \,dk,      \label{Fourier}
\end{equation}
where $\alpha(k)$ is the Fourier transform of $\eta(x,0)$. Writing the dispersion relation $\Omega(k)$ as a power series (Taylor expansion) about a fixed wavenumber $k_{0}$ and neglecting $\cal{O}$$(\kappa^{3})$ terms, we find
\begin{equation}
\beta\,\kappa^{2} = -\Omega(k_{0}+\kappa) + \Omega(k_{0}) + \Omega'\!(k_{0})\,\kappa,
\end{equation}
where $\beta = -\frac{1}{2}\,\Omega^{''}\!(k_{0})$. Let us define $k = k_{0} + \kappa$, $\tau = t$, and $\xi = x - \Omega'\!(k_{0})\,t$. Then equation~(\ref{Fourier}) can be written like
\begin{equation}
\eta(x,t) = e^{\,i[k_{0}\,x - \Omega(k_{0})\,t]}\!\! \int^{\,\infty}_{\!\!-\infty}\!\!\alpha(k_{0} + \kappa) \:e^{\,i\,\kappa\,\xi}\,e^{i\,\beta\,\kappa^{2}\,\tau} \,d\kappa. \label{eta}
\end{equation}
Denoting the integral in~(\ref{eta}) with $\psi(\xi,\tau)$, we find that $\psi(\xi,\tau)$ satisfies
\begin{equation}
\frac{\partial \psi}{\partial \tau} + i\,\beta\,\frac{\partial^{2} \psi}{\partial \xi^{2}} = 0. \label{linschro}
\end{equation}
This is the \textit{linear Schr\"{o}dinger} equation for \textit{narrow--banded} spectra. Equation~(\ref{linschro}) is a partial differential equation that describes the time evolution of the envelope of a linear wave packet \cite{Scot99}. Equation~(\ref{linschro}) has a monochromatic mode solution $\psi(\xi,\tau) = e^{i\,(\kappa\, \xi\, - \,\nu\, \tau )}$ where $\nu = - \beta\,\kappa^{2}$.

\subsection{Nonlinear Theory}

Assuming narrow-banded spectra, we consider the wave elevation in the following form $\eta(x,t) = \psi(\xi,\tau)\, e^{i(k_{0} x - \omega_{0} t)} + c.c.$, where $\tau = t$, $\xi = x - \Omega'\!(k_{0})\,t$, $\psi(\xi,\tau)$ is a complex-valued function (called the \textit{complex amplitude}), $k_{0}$ and $\omega_{0}$ are central wavenumber and frequency, respectively. The `physical' variables $x$ and $t$ represent the spatial and temporal variables, respectively. The evolution of the wave elevation is a weakly nonlinear deformation of a nearly harmonic wave with the fixed wavenumber $k_{0}$. If we substitute $\eta$ to the equations which describe the physical motion of the water waves (see below), then one finds that the complex amplitude $\psi(\xi,\tau)$ satisfies the \textbf{\textit{nonlinear Schr\"{o}dinger (NLS) equation}}.

As an example, the NLS equation can be derived from the modified KdV equation $\eta_{t} + i\,\Omega\,(-i\partial_{x})\eta + \frac{3}{4}\,\partial_{x}\eta^2 = 0$~\cite{Cahy02}. With $\tau = t$ and $\xi = x - \Omega'\!(k_{0})\,t$, the corresponding NLS equation\footnote[1]{In our paper \cite{Karj02}, we used the NLS equation with a different factor $i = \sqrt{-1}$. Therefore, the coefficients $\beta$ and $\gamma$ here have the opposite signs compared to the corresponding coefficients in \cite{Karj02}.} reads
\begin{equation}
\bigskip \frac{\partial \psi}{\partial \tau} + i\beta\,\frac{\partial ^{2}\psi}{\partial \xi^{2}} + i\gamma\,\left| \psi\right| ^{2}\psi = 0, \quad \beta, \gamma \in \mathbb{R}, \label{generalNLS}
\end{equation}
where $\beta$ and $\gamma$ depend only on $k_{0}$:
\begin{eqnarray}
\beta = -\frac{1}{2}\,\Omega^{''}\!\!(k_{0}), \label{beta}
\end{eqnarray}
\begin{equation}
\gamma = \frac{9}{4}\,k_{0}\,\left(\frac{1}{\Omega^{'}\!(k_{0})
- \Omega^{'}\!(0)} + \frac{k_{0}}{2\,\Omega(k_{0}) -
\Omega(2k_{0})} \right), \label{gamma}
\end{equation}
where $\omega = \Omega(k)$ is the linear dispersion relation~\cite{Groe98}. This equation arose as a model for packets of waves on deep water.

There are two types of NLS equations:
\begin{itemize}
\item If $\beta$ and $\gamma$ have the same sign, i.e., $\beta\,\gamma > 0$, then~(\ref{generalNLS}) is called the \textit{focusing} NLS equation (an attractive nonlinearity, modulationally unstable [Benjamin--Feir instability])~\cite{Sule99}.

\item If $\beta$ and $\gamma$ have a different sign, i.e., $\beta\,\gamma < 0$, then~(\ref{generalNLS}) is called the \textit{de-focusing} NLS equation (a repulsive nonlinearity, stable solution)~\cite{Sule99}.
\end{itemize}

The wavenumber $k_{\textmd{\scriptsize{crit}}}$, for which $\beta\,\gamma = 0$ holds, is called the \textit{critical wavenumber} or the \textit{Davey--Stewartson value}. Using the physical quantity $g = 9.8 $ m/s$^{2}$ and the water depth $h = 5$~m, then for wavenumbers $k_{0} > k_{\textmd{\scriptsize{crit}}} = 0.23$, the product $\beta\,\gamma > 0$, and the NLS equation is of focusing type. The corresponding critical wavelength is 27.41~m.

In the following, we consider only the focusing NLS equation. The NLS equation~(\ref{generalNLS}) has a \textit{plane--wave} solution
\begin{equation}
A(\tau) = r_{0}\,e^{\,-i\,\gamma\,r_{0}^{2}\,\tau}.			\label{solutiondependontime}
\end{equation}
In physical variables, the surface wave elevation is given as $\eta(x,t) = 2\,r_{0} \cos\,(k_{0}x - \omega_{0}t - \gamma r_{0}^2 t)$. Note that the corresponding phase velocity is $\frac{\omega_{0} + \gamma\,r_{0}^2}{k_{0}}$, where $\omega = \omega_{0} + \gamma\,r_{0}^2$ is a special case of the `nonlinear dispersion relation'. \label{7} In the following section, we analyze the stability of this \textit{plane--wave} solution.

\section{Benjamin--Feir Instability} \label{BF}

To investigate the instability of the NLS plane--wave solution, we perturb the $\xi$--independent function $A(\tau)$ with a small perturbation of the form $\epsilon(\xi,\tau) = A(\tau)\,B(\xi,\tau)$. We look for the cases where under a small perturbation the amplitude of the plane--wave solution grows in time~\cite{Debn94}:
\begin{equation}
\psi(\xi,\tau) = A(\tau)[1 + B(\xi,\tau)]. \label{perturbation}
\end{equation}
Substituting~(\ref{perturbation}) into~(\ref{generalNLS}), and ignoring nonlinear terms we obtain the \textit{linearized} NLS equation:
\begin{equation}
B_{\tau} + i \beta\,B_{\xi\xi} + i \gamma\,r_{0}^{2}(B + B^{*}) = 0.		\label{linearNLS}
\end{equation}
We seek the solutions of (\ref{linearNLS}) in the form
\begin{equation}
B(\xi,\tau) = B_{1}e^{(\sigma \tau + i \kappa \xi)} + B_{2}e^{(\sigma^{*}\tau - i \kappa \xi)},			\label{perturbationsolution}
\end{equation}
where $B_{1}, B_{2} \in \mathbb{C}$, $k_{0} + \kappa$ is the local wavenumber, $\kappa$ is perturbation wavenumber\footnote[1]{In our paper~\cite{Karj02}, we also called the \textit{perturbation} wavenumber as the \textit{modulation} wavenumber.}, and $\sigma \in \mathbb{C}$ is called the \textit{growth rate}. If Re\,($\sigma) > 0$, then the perturbed solution of the NLS equation grows exponentially. This is the criterion for the so-called \textbf{Benjamin--Feir instability} of a one--wave mode with perturbation wavenumber~$\kappa$~\cite{Debn94}.

Substituting the function $B(\xi,\tau)$ in~(\ref{perturbationsolution}) into~(\ref{linearNLS}) yields a pair of coupled equations that can be written in matrix form as follows
\begin{equation}
\left(
\begin{array}{cc}
\sigma - i\beta\,\kappa^{2} + i\gamma\,r_{0}^{2}  &  i\gamma\,r_{0}^{2} \\
-i\gamma\,r_{0}^{2} & \sigma + i \beta\,\kappa^{2} - i\gamma\,r_{0}^{2}
\end{array}
\right)
\left(
\begin{array}{c}
B_{1} \\ B_{2}^{*}
\end{array}
\right) =
\left(
\begin{array}{c}
0 \\ 0
\end{array}
\right). \label{matrix}
\end{equation}
A nontrivial solution to~(\ref{matrix}) can exist only if the determinant of the left-hand side matrix is zero. This condition reads as follows
\begin{equation}
\sigma^2 = \beta\,\kappa^{2}(2\,\gamma\,r_{0}^{2} - \beta\,\kappa^{2}).
\end{equation}
We have the following cases:
\begin{itemize}
\item The growth rate $\sigma$ is real and positive if $\kappa^{2} < 2\frac{\gamma}{\beta}\,r_{0}^{2}$. This corresponds to \textbf{Benjamin--Feir instability}. For these values of $\kappa$, the perturbation amplitude is exponentially amplified in time~\cite{Debn94}. 

\item The growth rate $\sigma$ is purely imaginary if $\kappa^{2} > 2\frac{\gamma}{\beta}\,r_{0}^{2}$. This corresponds to a plane--wave solution that has a bound amplitude for all time~\cite{Debn94}.
\end{itemize}
Thus, the range of instability is given by
\begin{equation}
0 < |\kappa| < \left|\kappa_{\textmd{\scriptsize{crit}}}\right| = \sqrt{\frac{2\,\gamma}{\beta}}\,r_{0}.
\end{equation}
It is easy to find that the `strongest' instability occurs at $\kappa_{\textmd{\scriptsize{max}}} = \sqrt{\frac{\gamma}{\beta}}\,r_{0},$ where the maximum growth rate is $\sigma_{\textmd{\scriptsize{max}}} = \gamma\,r_{0}^2$. Figure~\ref{plot_growth_rate} shows the plot of growth rate $\sigma$ as a function of perturbation wavenumber $\kappa$ for $r_{0} = \beta = \gamma = 1$.
\begin{figure}[h]
\centering
\includegraphics[width = 0.5\textwidth]{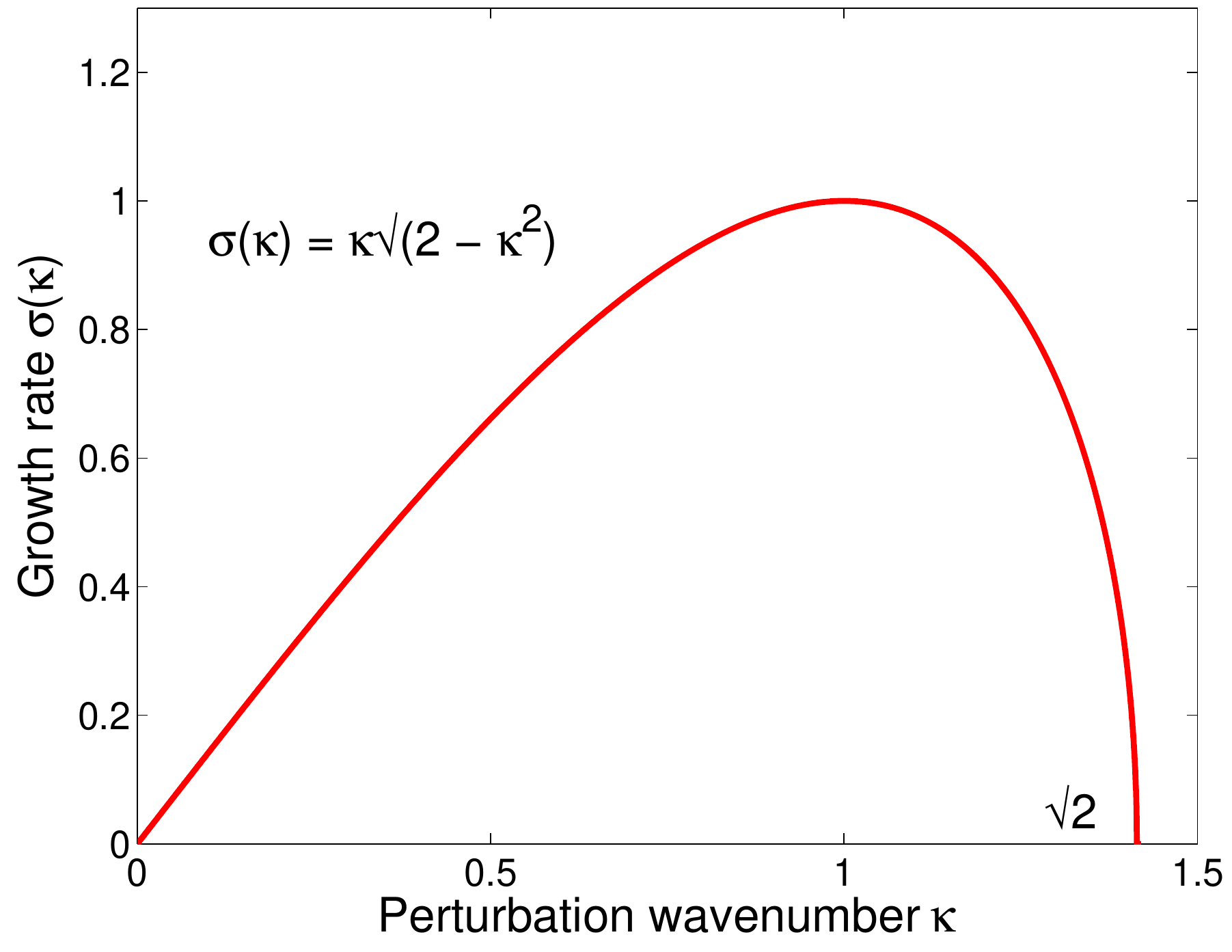}
\caption{Growth rate $\sigma$ as a function of perturbation wavenumber $\kappa$ for $r_{0} = \beta = \gamma = 1$.}	   \label{plot_growth_rate}
\end{figure}

We can write the solution of the linearized NLS equation as
\begin{equation}
\psi(\xi,\tau) = r_{0}\,e^{\,-i\,\gamma\,r_{0}^2\,\tau} [ 1 + e^{\sigma(\kappa)\,\tau}(B_{1} e^{i\,\kappa\,\xi} + B_{2} e^{-i\kappa\,\xi})],
\end{equation}
with $\sigma(\kappa) = \kappa\,\sqrt{2\,\beta\,\gamma\,r_{0}^{2} - \beta^{2}\,\kappa^{2}}$. Since this solution is obtained from the linearized NLS equation, it is only valid if amplitudes are small. When the time is increased, the amplitude increases exponentially and the linearized theory becomes invalid. Therefore, we cannot use the solution from the linearized equation to investigate the behavior of the maximum amplitude increase due to the Benjamin--Feir instability. Fortunately, there exists an exact solution to the NLS equation that describes the exact behavior of the wave profile and corresponds to the Benjamin--Feir instability.

\section{Soliton on Finite Background}

In this section, we investigate the relation between the maximum amplitude of a certain solution of the NLS equation and the perturbation wavenumber $\kappa$ in the instability interval. Suppose that $\widetilde{\psi}(\widetilde{\xi},\widetilde{\tau})$ satisfies the NLS equation, then by scaling the variables $\widetilde{\xi}$ and $\widetilde{\tau}$ and also $\widetilde{\psi}$ itself using the relation
\begin{equation}
\widetilde{\psi}(\widetilde{\xi},\widetilde{\tau}) = \frac{1}{r_{0}}\, \psi\left(\xi,\,\tau\right),
\end{equation}
we get the normalized NLS equation for $\widetilde{\psi}$. An exact solution, the so--called \textit{soliton on finite background} (SFB)\footnote[1]{This SFB sometimes also called the \textit{second most important solution} of the NLS equation. Another author~\cite{Osbo01} calls it the \textit{rogue wave solution}.} of the NLS equation is given by~\cite{Akhm97} :
\begin{equation}
\psi(\xi,\tau) := \frac{(\kappa^{2} - 1) \cosh (\sigma(\kappa)\,\tau) + \sqrt{\frac{2-\kappa^{2}}{2}}\cos (\kappa \xi) - i \sigma(\kappa) \sinh (\sigma(\kappa)\,\tau)} {\cosh (\sigma(\kappa)\,\tau) - \sqrt{\frac{2-\kappa^{2}}{2}} \cos (\kappa \xi)}\; e^{\,-i\,\tau}, \label{exact2}
\end{equation}
where $0 < \kappa < \sqrt{2}$ and $\sigma(\kappa) = \kappa \sqrt{2-\kappa^{2}}$. The formula for the SFB in non-normalized form can be found in Appendix~\ref{A}. \\
Now consider the SFB in normalized form~(\ref{exact2}). For $\kappa = 1$, we have
\begin{equation}
\psi(\xi,\tau) := \frac{\cos \xi - i \sqrt{2} \sinh \tau }{\sqrt{2} \cosh \tau - \cos \xi}\;e^{\,-i\,\tau}. 	\label{exact1}
\end{equation}
Figure~\ref{plot_SFB1} shows the plot of $\left|\psi \right|$ from~(\ref{exact1}) as a function of $\xi$ and $\tau$. Note that $\psi(\xi,\tau)$ is a $2\pi$--periodic function with respect to $\xi$ variable.
\begin{figure}[h]
\centering 
\includegraphics[width = 0.7\textwidth]{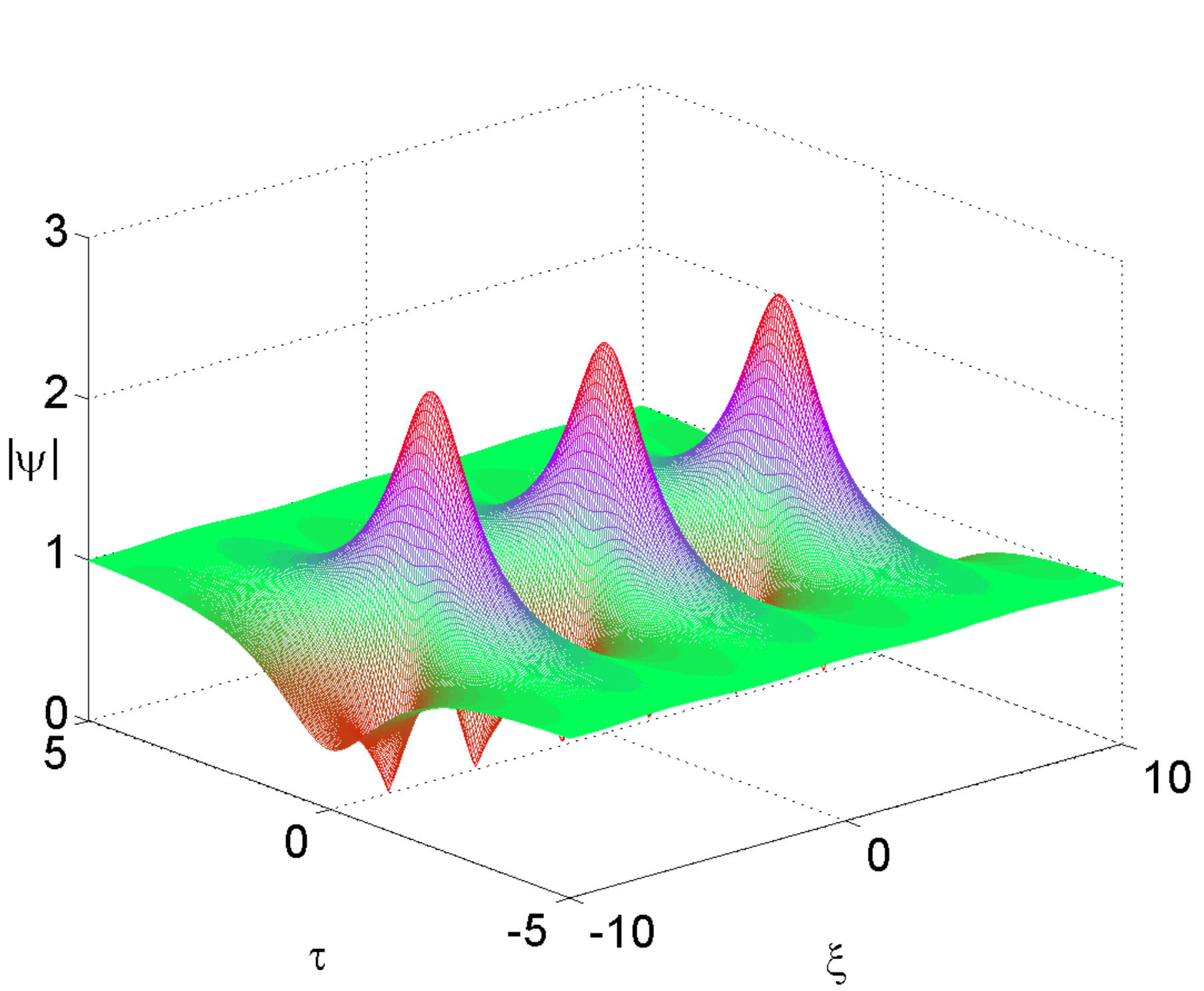}
\caption{Plot of $\left|\psi \right|$ as a function of $\xi$ and $\tau$ for perturbation wavenumber $\kappa = 1.$} 			  \label{plot_SFB1}
\end{figure}

In the following, we analyze the solution~(\ref{exact2}) in detail. The behavior for this SFB as $\tau \rightarrow \pm \infty$ is given by
\begin{equation}
\lim_{\tau \rightarrow \infty} \; \left|\psi(\xi,\tau)\right| =
\sqrt{(\kappa^2 -1)^2 + \kappa^2\,(2 - \kappa^2)} = 1 .
\end{equation}
Because of this property, the solution~(\ref{exact2}) is called as SFB. For a `normal' soliton, the elevation vanishes at infinity: the `normal' soliton is exponentially confined. For SFB, the solution is a similar elevation on top of the finite (nonzero) background level, here the normalized value~1.

Write the solution in the form $\psi(\xi,\tau) = u(\xi,\tau)\,e^{\,-i\,\tau}$, where $u(\xi,\tau)$ where $u$ describes the amplitude and the exponential part expresses oscillations in time. Let us investigate the behavior of $u(\xi,\tau)$ in time. For that consider the limiting behavior of $\partial_{\tau}u$ as $\tau \rightarrow \pm \infty$. We have
\begin{equation}
\frac{\partial u}{\partial \tau}(\xi,\tau) = \frac{-i\sigma^{2}(\kappa) - \sigma(\kappa)\,\sqrt{\frac{2-\kappa^2}{2}} \cos (\kappa \, \xi)\left[\kappa^2 \,\sinh \left(\sigma(\kappa)\,\tau \right) - i\,\sigma(\kappa)\,\cosh \left(\sigma(\kappa)\,\tau\right) \right]}{\left[\cosh \left(\sigma(\kappa)\,\tau\right) - \sqrt{\frac{2-\kappa^2}{2}} \cos \left(\kappa\,\xi\right)\right]^{2}}.
\end{equation}
If $\xi \neq \frac{\pi}{2\,\kappa}$, then
\begin{equation}
\frac{\partial u}{\partial \tau} \approx - 2\,(\kappa^2 - i\,\sigma(\kappa)) e^{-\sigma(\kappa)\,\tau} \quad \textmd{if} \; \tau \gg 0,
\end{equation}
and
\begin{equation}
\frac{\partial u}{\partial \tau} \approx 2\,(\kappa^2 + i\,\sigma(\kappa)) e^{\,\sigma(\kappa)\,\tau} \quad \textmd{if} \; \tau \ll 0.
\end{equation}
If $\xi = \frac{\pi}{2\,\kappa}$, then
\begin{equation}
\frac{\partial u}{\partial \tau} \approx -4\,i\,\sigma^2(\kappa) \,e^{-2\,\sigma(\kappa)\,\tau} \quad \textmd{for} \; \tau \gg 0,
\end{equation}
and
\begin{equation}
\frac{\partial u}{\partial \tau} \approx -4\,i\,\sigma^2(\kappa)\, e^{\,2\,\sigma(\kappa)\,\tau} \quad \textmd{for} \; \tau \ll 0.
\end{equation}

\begin{figure}[h]
\centering 
\includegraphics[width = 0.7\textwidth]{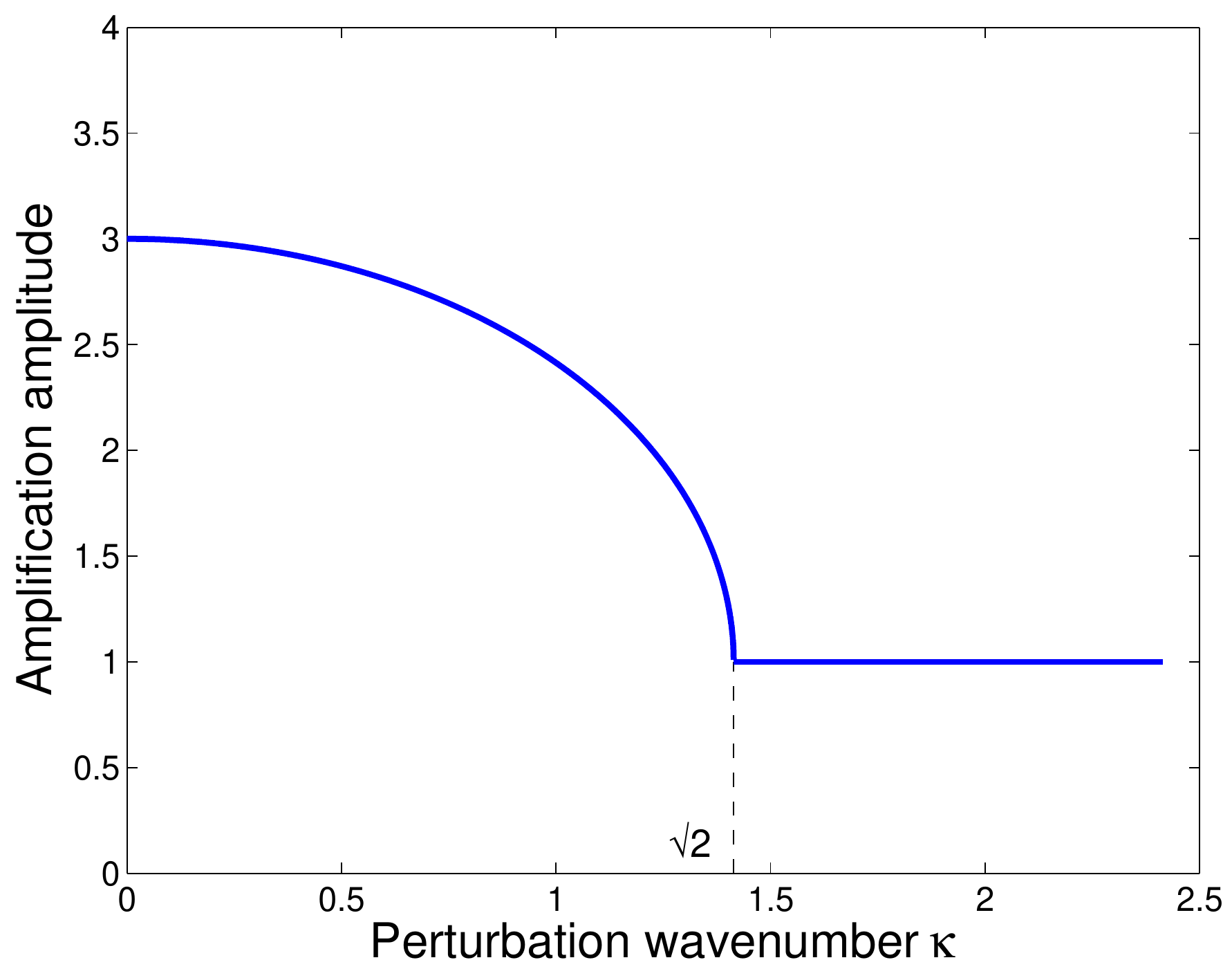}
\caption{The plot of amplitude amplification as a function of the perturbation wavenumber~$\kappa$.} 		\label{plot_amplification_amplitude}
\end{figure}
Next, let us find the relation between the maximum amplitude of the exact solution (\ref{exact2}) with perturbation wavenumber $\kappa$, where $0 < \kappa < \sqrt {2}$. The maximum value of the complex amplitude is at $\xi \equiv 0\;(\textmd{mod}\;2\pi)$ and when $\tau = 0$. So, we have
\begin{equation}
|\psi|_{\textmd{\scriptsize{max}}} = |\psi(0,0\,;\,\kappa)| = \frac{\kappa^{2} - 1 + \sqrt{1 - \frac{1}{2} \kappa^{2}}}{1 - \sqrt{1 - \frac{1}{2} \kappa^{2}}}.
\end{equation}
Using the approximation $\sqrt{1+a} \approx 1 + \frac{1}{2} a$ for small $a$, and apply it to our formula ($a = - \frac{1}{2} \kappa^{2}$), we obtain
\begin{eqnarray}
\lim_{\kappa \rightarrow 0}\;|\psi(0,0;\kappa)| = 3.
\end{eqnarray}
As a result, the maximum factor of the amplitude amplification is
\begin{equation}
\lim_{\kappa \rightarrow 0}\;\frac{|\psi(0,0;\kappa)|}{\lim_{\tau \rightarrow \pm \infty}|\psi(\xi,\tau;\kappa)|} = \frac{3}{1} = 3.
\end{equation}
Figure~\ref{plot_amplification_amplitude} shows the plot of the maximum amplitude $|\psi|_{\textmd{\scriptsize{max}}}$ as a function of perturbation wavenumber~$\kappa$.

To summarize, Figure~\ref{plotgabung} compares the plot of the dispersion relation $\Omega$ and its quadratic approximation (leads to the NLS equation), the growth rate $\sigma$ (which is related to Benjamin--Feir instability), and the maximum amplitude of $\psi$ as function of wavenumber $k$.\\
\begin{figure}[h]
\centering
\includegraphics[width = 0.7\textwidth]{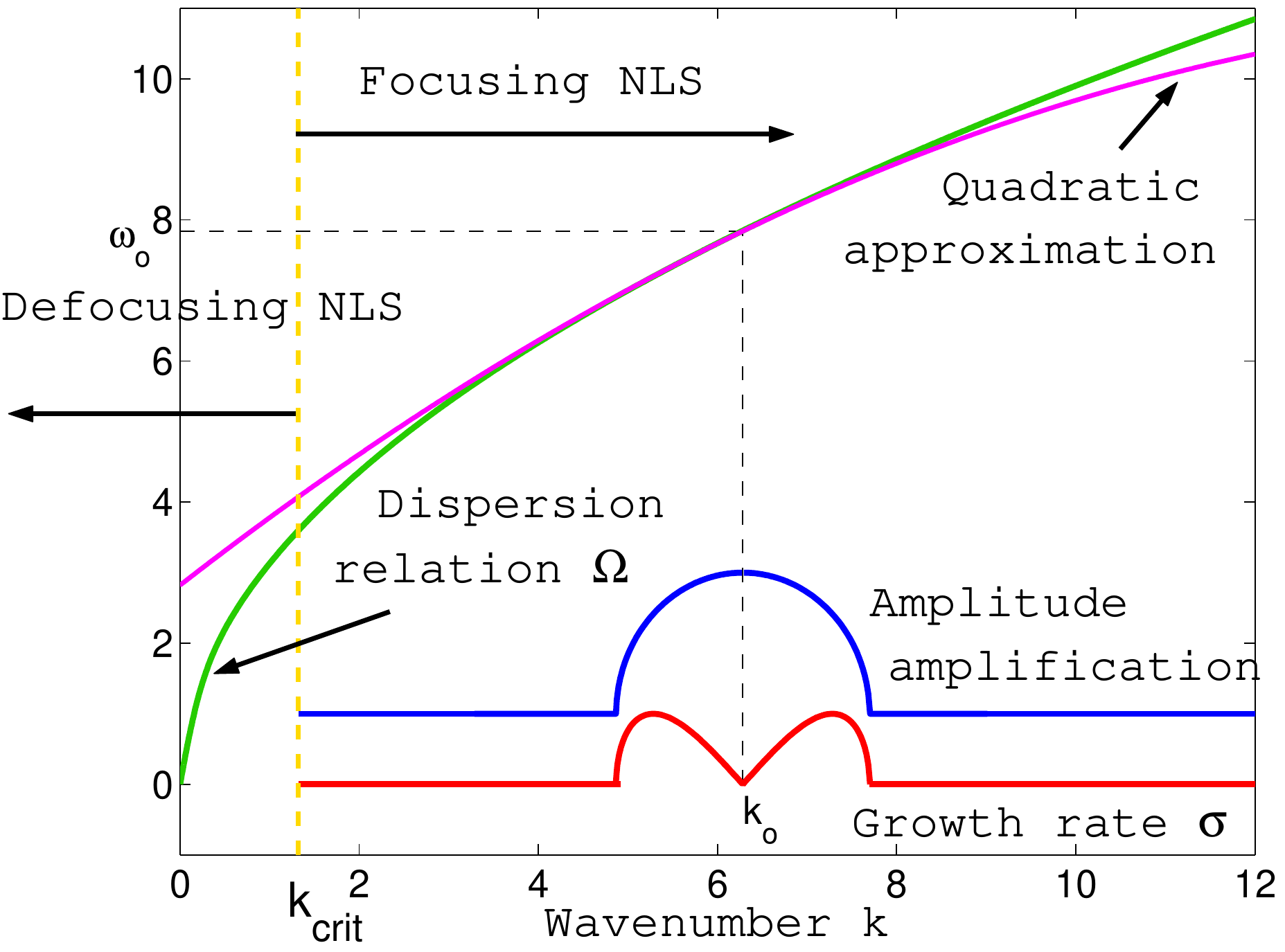}
\caption{Dispersion relation $\Omega$ and its quadratic approximation, growth rate~$\sigma$ (Benjamin--Feir instability), and maximum amplitude of $\psi$ as functions of wavenumber~$k$. The case corresponds to $k_{0} = 2\,\pi$.} 							\label{plotgabung}
\end{figure}

\chapter{Wave Dislocation and Phase Singularity} \label{WD_PS}

\section{Wave Dislocation and Phase Singularity}

Consider the NLS equation~(\ref{generalNLS}) in the moving frame of reference. By transforming into the original spatial and temporal variables, we have the NLS equation in the non-moving frame of reference:
\begin{equation}
\bigskip \frac{\partial \psi}{\partial t} + V_{0}\,\frac{\partial \psi}{\partial x} + i\,\beta\,\frac{\partial^{2}\psi}{\partial x^{2}} + i\,\gamma\,\left| \psi\right| ^{2}\psi = 0, \qquad \qquad V_{0}, \beta, \gamma \in \mathbb{R}, 			\label{NLS}
\end{equation}
where $V_{0} = \Omega^{'}\!(k_{0})$, $\beta$ and $\gamma$ as defined in~(\ref{beta}) and~(\ref{gamma}), respectively. We write the solution of the NLS equation in the polar form $\psi(x,t) = a(x,t)\,e^{i\,\theta(x,t)}$, where $a(x,t)$ is the \textit{real-valued amplitude} and  $\theta(x,t)$ is the \textit{real-valued phase}. The physical wave profile is given by the elevation of the surface wave
\begin{eqnarray}
\eta(x,t) &=& \psi(x,t)\,e^{\,i\,(k_{0}x - \Omega(k_{0}) t)} + \textmd{c.c. \,(complex conjugate)} \nonumber \\
          &=& 2\,a(x,t)\,\cos \Phi(x,t),
\end{eqnarray}
where $\Omega(k)$ is the \textit{linear dispersion relation} and $\Phi(x,t) = \theta(x,t) + k_{0}x - \omega_{0} t$ is the \textit{real-valued phase} related to the physical wave profile.

We define the local wavenumber $k(x,t)$ and local frequency $\omega(x,t)$ related to the phase $\Phi(x,t)$:
\begin{eqnarray}
k(x,t) &=& \frac{\partial \Phi}{\partial x} = \frac{\partial \theta}{\partial x} + k_{0},  						\label{local_wavenumber} \\
\omega(x,t) &=& - \frac{\partial \Phi}{\partial t} = - \frac{\partial \theta}{\partial t} + \omega_{0}. 		\label{local_frequency}
\end{eqnarray}
The complete expressions for these quantities can be found in Appendix~\ref{A}.

Figure~\ref{physicalSFB} shows the plots of two physical wave profiles corresponding to SFB of the NLS equation for two distinct perturbation wavenumbers~$\kappa$.
\begin{figure}[h]
\begin{center}
\hspace*{-0.1cm}
\includegraphics[width = 0.45\textwidth]{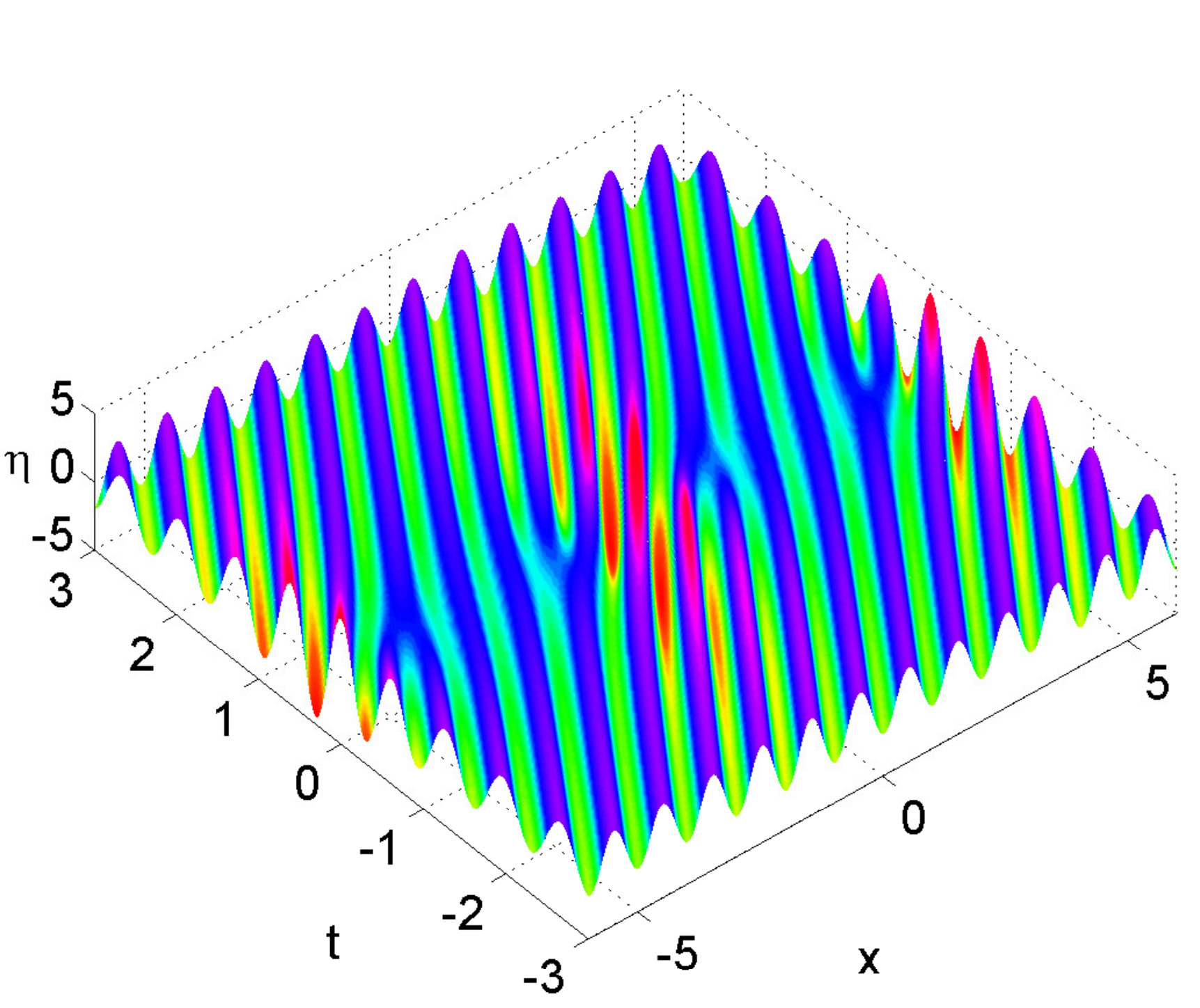}		\hspace{0.5cm}
\includegraphics[width = 0.45\textwidth]{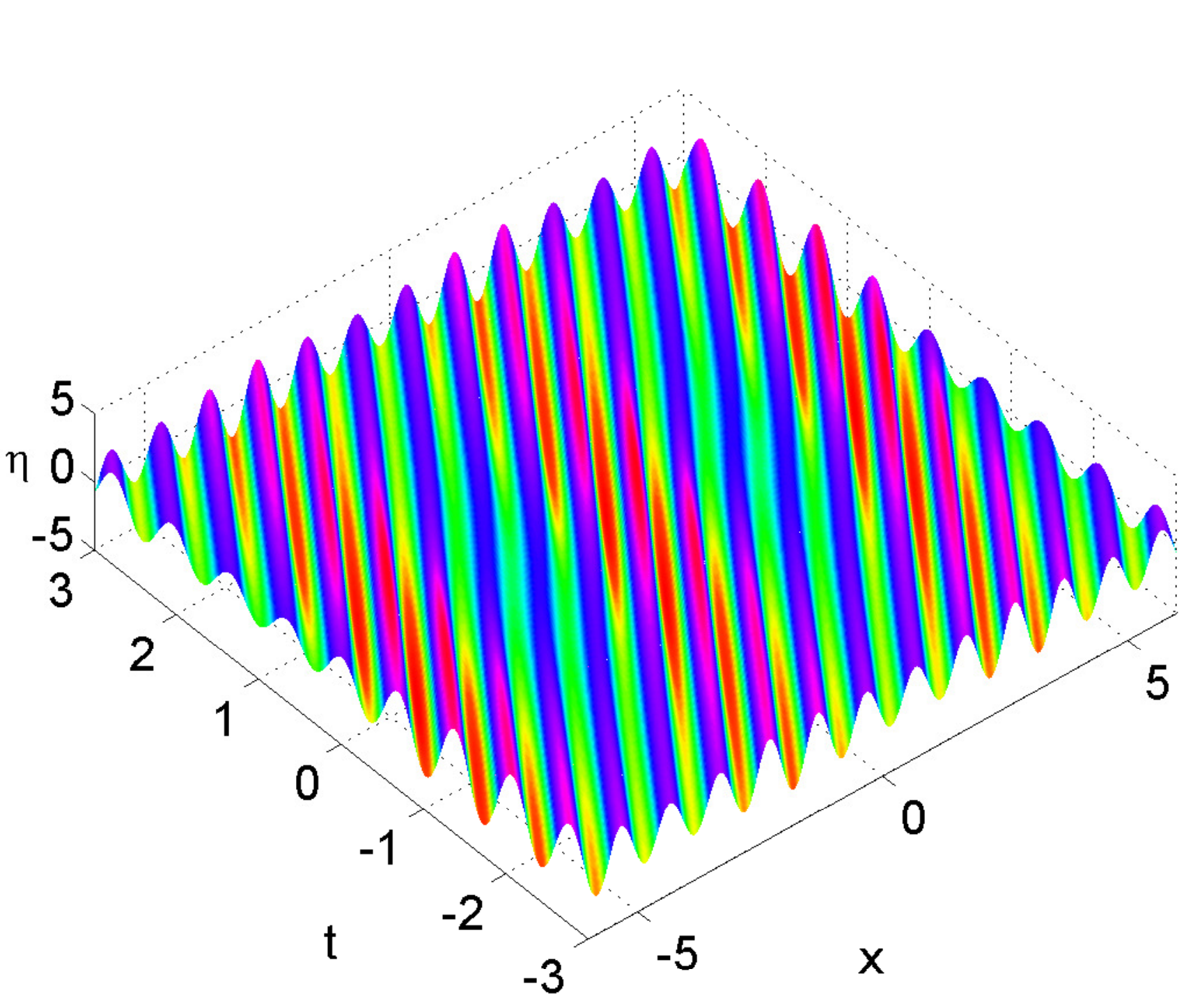}
\end{center}
\caption{Three-dimensional plots of the physical wave group related to the SFB of NLS for perturbation wavenumber $\kappa = 1$ (left) and for $\kappa = 1.33$ (right). The physical wave profile for $0 < \kappa \leq \kappa_{\textmd{\scriptsize crit}}$ has \textit{wave dislocation} phenomenon. There is no wave dislocation when $\kappa_{\textmd{\scriptsize crit}} < \kappa < \sqrt{2}$.}  \label{physicalSFB}
\end{figure}

The critical perturbation wavenumber, $\kappa_{\textmd{\footnotesize crit}}$, is defined as a value when the real-valued amplitude has zeros for the last time when $\kappa$ is increasing. It is found that the real-valued amplitude has zeros for $0 < \kappa \leq \kappa_{\textmd{\footnotesize crit}}$ and strictly positive for perturbation wavenumber $\kappa_{\textmd{\footnotesize crit}} < \kappa < \sqrt{2}$.\footnote[1]{The same observation also can be found in \textit{steepness}, defined as the multiplication of real-valued amplitude $a$ and the local wavenumber $k$. Dr. Andonowati and her PhD students in Bandung, Indonesia found the same phenomenon in the steepness.}
\begin{proposition} \nonumber
The critical perturbation wavenumber $\kappa_{\textmd{\footnotesize crit}} =  \sqrt{\frac{3}{2}}$.
\end{proposition}
\begin{proof} \\
In this proof we will derive the value for $\kappa_{\textmd{\footnotesize crit}}$. This critical value is attained when $a(x,0) = 0$ and at the same position the local minimum is also reached. Therefore, we have to find $x_{0}$ such that $a(x_{0},0) = 0$ and $\partial_{x}a(x_{0},0) = 0$. The real-valued amplitude $a$ at $t = 0$ reads
\begin{eqnarray}
a(x,0) = r_{0}\frac{\left|\kappa^2 - 1 + \sqrt{1 - \frac{1}{2}\kappa^2} \cos\left(\kappa\,r_{0}\,\sqrt{\frac{\gamma}{\beta}}\,x\right)\right|}{1 - \sqrt{1 - \frac{1}{2}\kappa^2} \cos\left(\kappa\,r_{0}\,\sqrt{\frac{\gamma}{\beta}}\,x\right)}.   \label{a(x,0)}
\end{eqnarray}
By writing the numerator part as $|f(x)| = \sqrt{[f(x)]^{2}}$ and the denominator part as $g(x)$, then the first derivative with respect to $x$ reads
\begin{equation}
\frac{\partial a}{\partial x}(x,0) = r_{0}\frac{f(x)\,[f'(x)g(x) - f(x)g'(x)]}{\sqrt{[f(x)]^2}\,g^2(x)},
\end{equation}
where $f(x) \neq 0$ and $g(x) \neq 0$. For $f(x) > 0$, the fraction $f(x)/|f(x)| = 1$; and for $f(x) < 0$, the fraction $f(x)/|f(x)| = -1$. Therefore, the first derivative above turns into
\begin{equation}
\frac{\partial a}{\partial x}(x,0) = \pm r_{0}\frac{f'(x)g(x) - f(x)g'(x)}{g^2(x)},
\end{equation}
To find the extremum means to find zeros of the first derivative of the steepness, which leads into $f'(x)g(x) = f(x)g'(x)$. By finding zeros for this relation, we have $\kappa^2\,\sin\left(\kappa\,r_{0}\,\sqrt{\frac{\gamma}{\beta}}\,x \right) = 0$, which implies
\begin{equation}
x_{0} = \frac{n\,\pi}{\kappa\,r_{0}}\,\sqrt{\frac{\beta}{\gamma}}, \quad  \quad n \in \mathbb{Z}.
\end{equation}
The minimum values are reached when $n$ is odd, and thus the zeros turn into
\begin{equation}
x_{0} = \frac{(2m+1)\,\pi}{\kappa\,r_{0}}\,\sqrt{\frac{\beta}{\gamma}}, \quad \quad m \in \mathbb{Z}.
\end{equation}
Now substitute $x_{0}$ into~(\ref{a(x,0)}) to obtain $a(x_{0},0) = 0$:
\begin{equation}
a(x_{0},0) = r_{0}\frac{\left|\kappa^2 - 1 + \sqrt{1 - \frac{1}{2}\kappa^2} \cos\left[(2m+1)\,x\right]\right|}{1 - \sqrt{1 - \frac{1}{2}\kappa^2} \cos\left[(2m+1)\,x\right]} = 0. 				\label{a(x_{0},0)}
\end{equation}
For $m \in \mathbb{Z}$, $\cos[(2m+1)\,x] = -1$, so the denominator in~(\ref{a(x_{0},0)}) is nonzero and remains to solve
\begin{equation}
\kappa^2 - 1 = \sqrt{1 - \frac{1}{2}\kappa^2}.
\end{equation}
By squaring the equation and solving the equation for $\kappa$, we obtain $\kappa_{\textmd{\footnotesize crit}} = \sqrt{\frac{3}{2}}$.
\end{proof}

The corresponding real-valued amplitude for the physical wave profile SFB of NLS with different perturbation wavenumber $\kappa$, including $\kappa_{\textmd{\footnotesize crit}}$ (the red plot) can be seen in Figure~\ref{steepness_plots}. Notice also that the longer perturbation wavelength (corresponding to a smaller value of $\kappa$), the higher real-valued amplitude we have, leads to greater amplitude amplification factor. The wave profile for short perturbation wavelength (corresponding to $\sqrt{3/2} < \kappa < \sqrt{2}$), does not give any wave dislocation but it is similar to a periodic wave train. The amplitude amplification factor for $\kappa_{\textmd{\footnotesize crit}}$ is~2, as visible in Figure~\ref{steepness_plots} (bottom left).
\begin{figure}[h]
\begin{center}
\hspace*{-0.1cm}
\includegraphics[width = 0.45\textwidth]{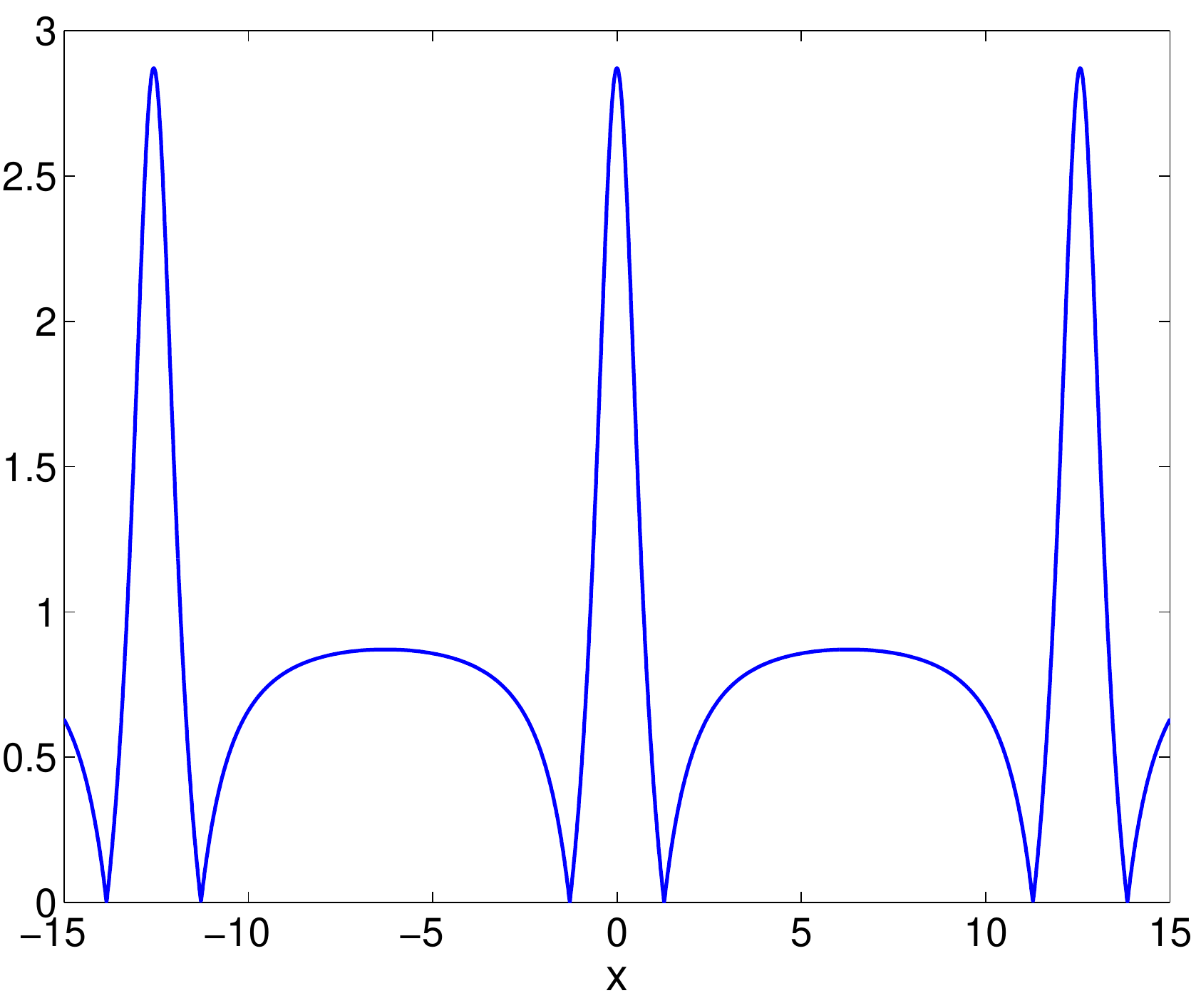}    	  	\hspace{0.5cm}
\includegraphics[width = 0.45\textwidth]{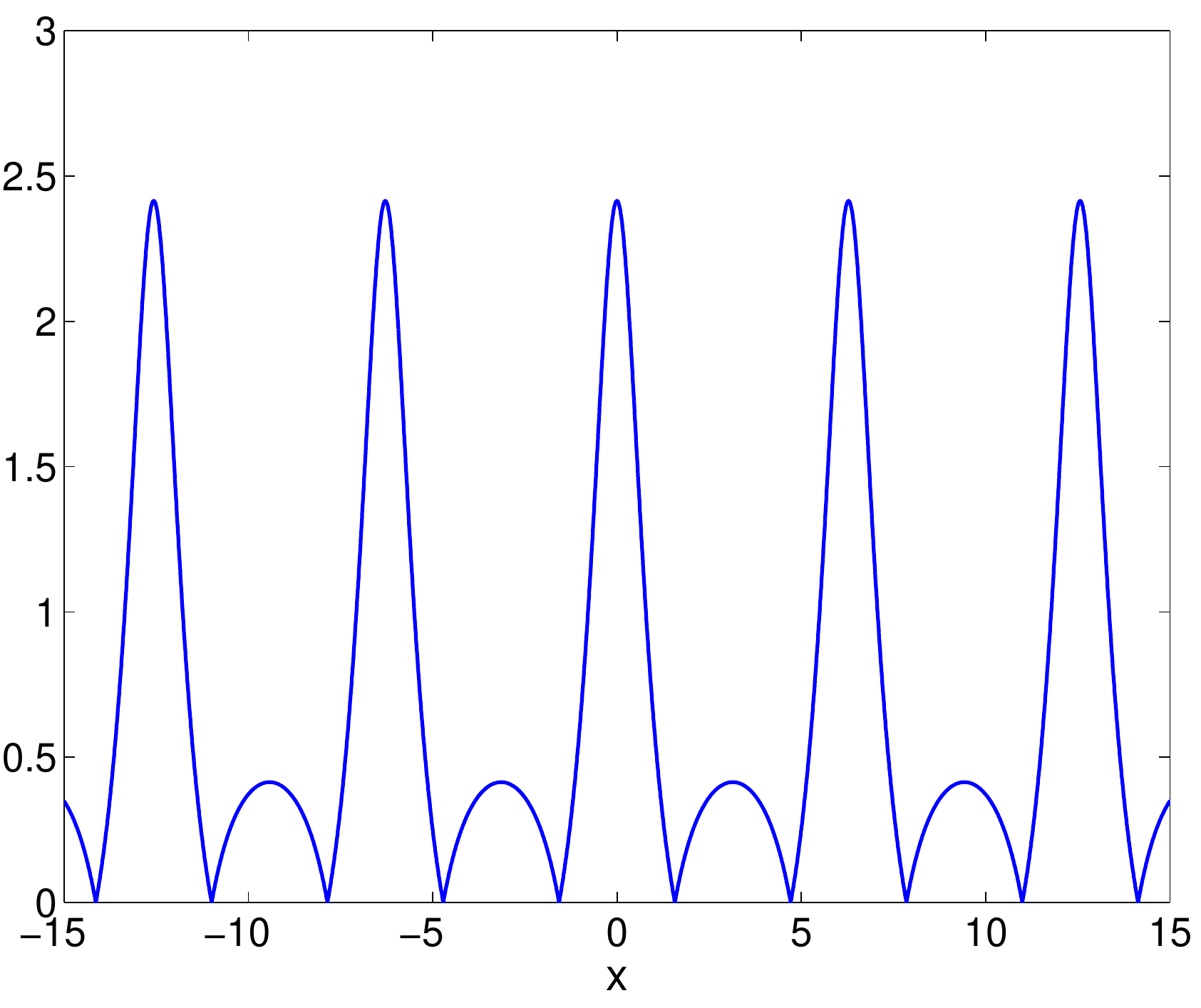}
\includegraphics[width = 0.45\textwidth]{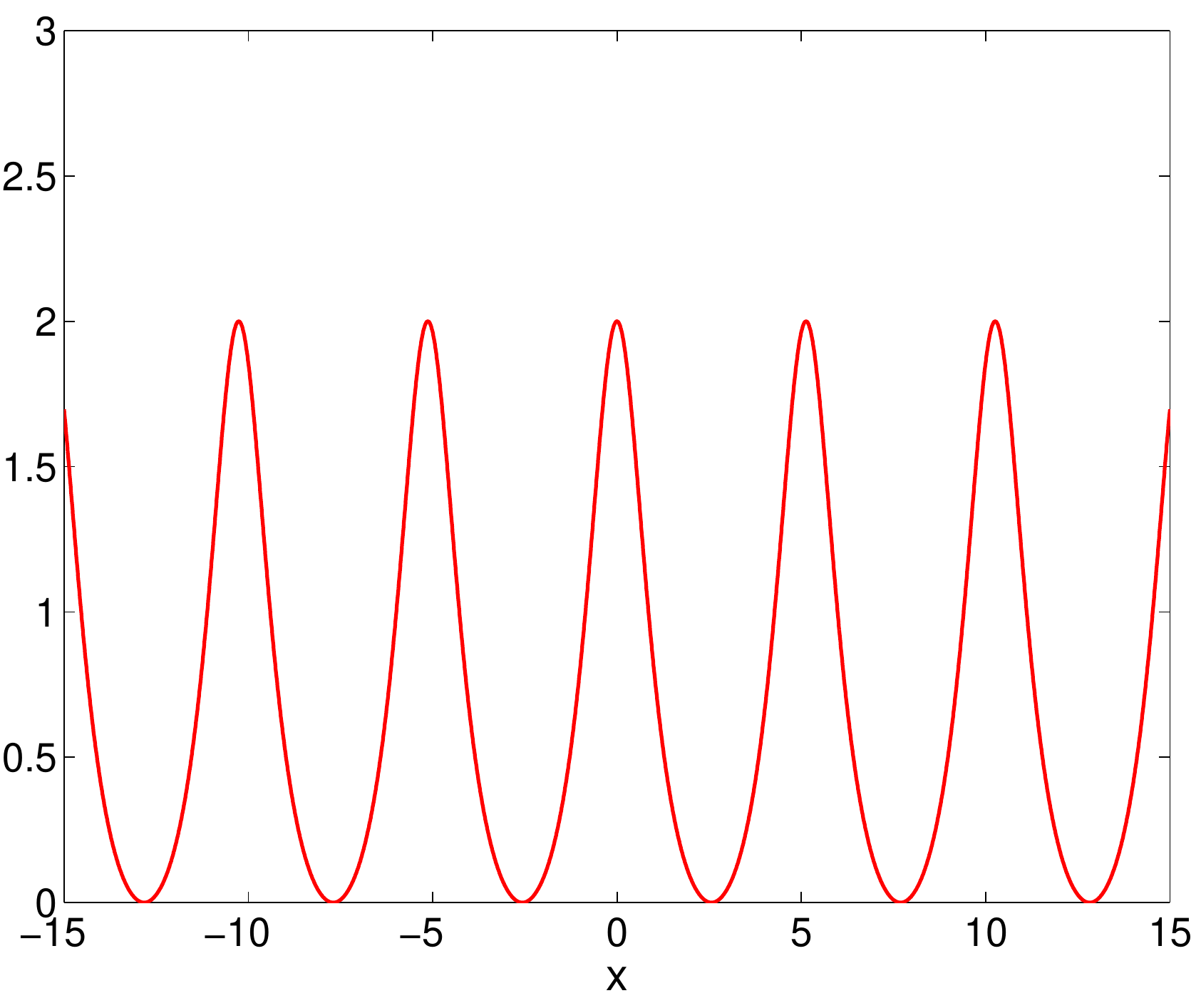}			\hspace{0.5cm}	
\includegraphics[width = 0.45\textwidth]{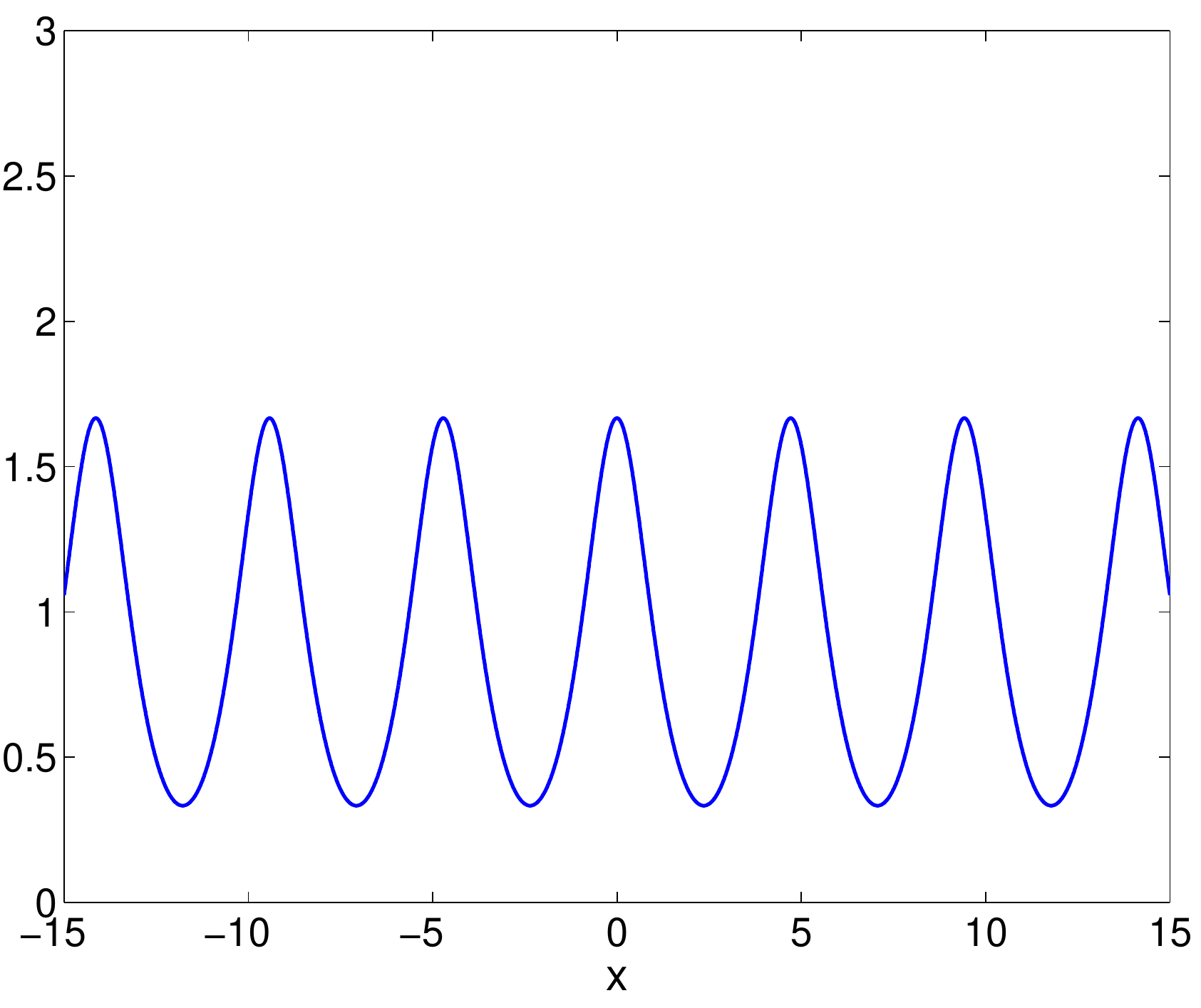}
\end{center}
\caption{Real-valued amplitude plots $a(x,0)$ of the physical wave group related to the SFB of NLS for perturbation wavenumber $\kappa = 0.5$ (top left) $\kappa = 1$ (top right), $\kappa_{\textmd{\footnotesize crit}} = \sqrt{3/2}$ (bottom left, red plot), and for $\kappa = 1.33$ (bottom right).}  				\label{steepness_plots}
\end{figure}

Now, observe again the plot of physical wave profile for $\kappa = 1$ in Figure~\ref{physicalSFB} (left). The plot illustrates what could be called \emph{wave compression caused by wave dislocation}. The (normalized) parameters are chosen such that the highest wave amplitude occurs at the center $x = t = 0$. The phenomenon is symmetric for mirroring through the center: $(x,t) \rightarrow (-x ,-t)$. First observe the wave dislocations at approximately $(x ,t) = (\pm \pi /2,0)$. This point is within an oval area of small amplitude waves; the oval approximately extends from $x = 1$ to $x = 5$, between $t = -1$ and $t = 1 $ (and mirrored). Then at the $x$-axis, in the interval $x \in (-1,1)$, large deformations are seen to form a steep wave with a single wave crest and two deep throughs. Along the $t$-axis, a series of increasingly larger amplitude (steeper) waves are seen when approaching the center from the negative time. Far away from the center, a regular pattern of straight phase fronts at each side of the $t$-axis represent a completely regular, periodic wave pattern; a detailed analysis in the next section (Figure~\ref{sfb_slices}) shows a phase shift of the pattern and periodic behavior in the $x$-direction.

The \textit{wave dislocation} phenomenon is happening when at a specific position and specific time two waves are merging into one wave or one wave is splitting into two waves. This phenomenon, as can be seen in Figure \ref{physicalSFB}, gives rise to large changes in the `local wavelength and frequency' which is best illustrated by introducing the phase-amplitude description and investigating the so-called \textit{dispersion plot}. With the definitions of local wavenumber (\ref{local_wavenumber}) and local frequency (\ref{local_frequency}), the evolution for these two quantities can be investigated in the dispersion plane of frequency versus wavenumber. In Figure \ref{trajectories} , for various instants of time, the trajectories are shown parameterized by $x$, $x\rightarrow (k(x,t),\omega (x,t))$ in the left picture, and in the right picture, the trajectories at fixed positions are shown parameterized by time $t\rightarrow (k(x,t),\omega (x,t))$. Observe the large deviations of the local wavenumber and local frequency, and the singular behavior that corresponds to the point $(\pi /2,0)$, justifying the \textit{phase singularity} phenomenon. It happens when the local wavenumber and local frequency become unbounded when the amplitude vanishes, as shown in the dispersion plane trajectories. (See Figure~\ref{trajectories}.) \newline
\begin{figure}[h]
\begin{center}
\includegraphics[width = 0.45\textwidth]{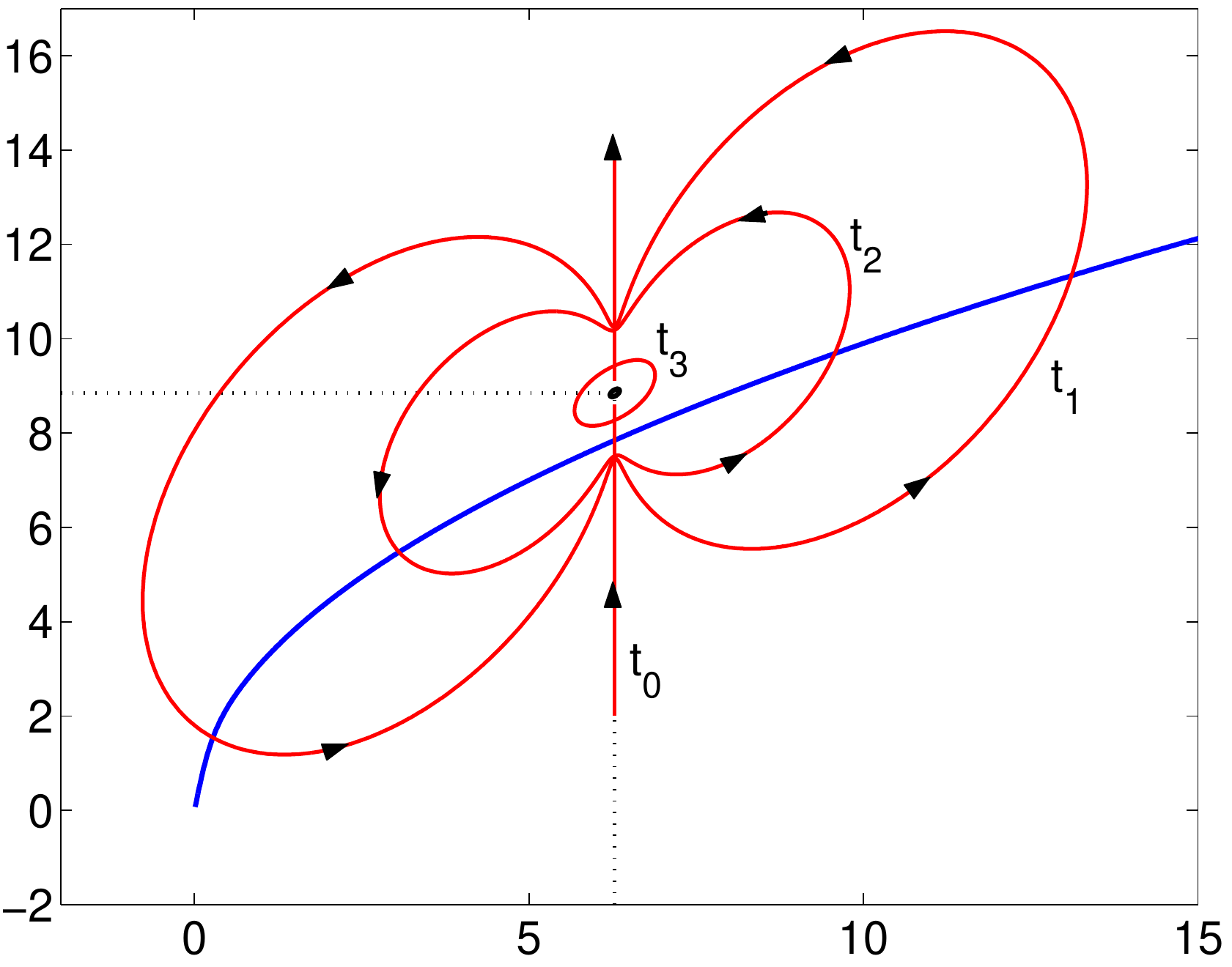}			\hspace{0.5cm}
\includegraphics[width = 0.45\textwidth]{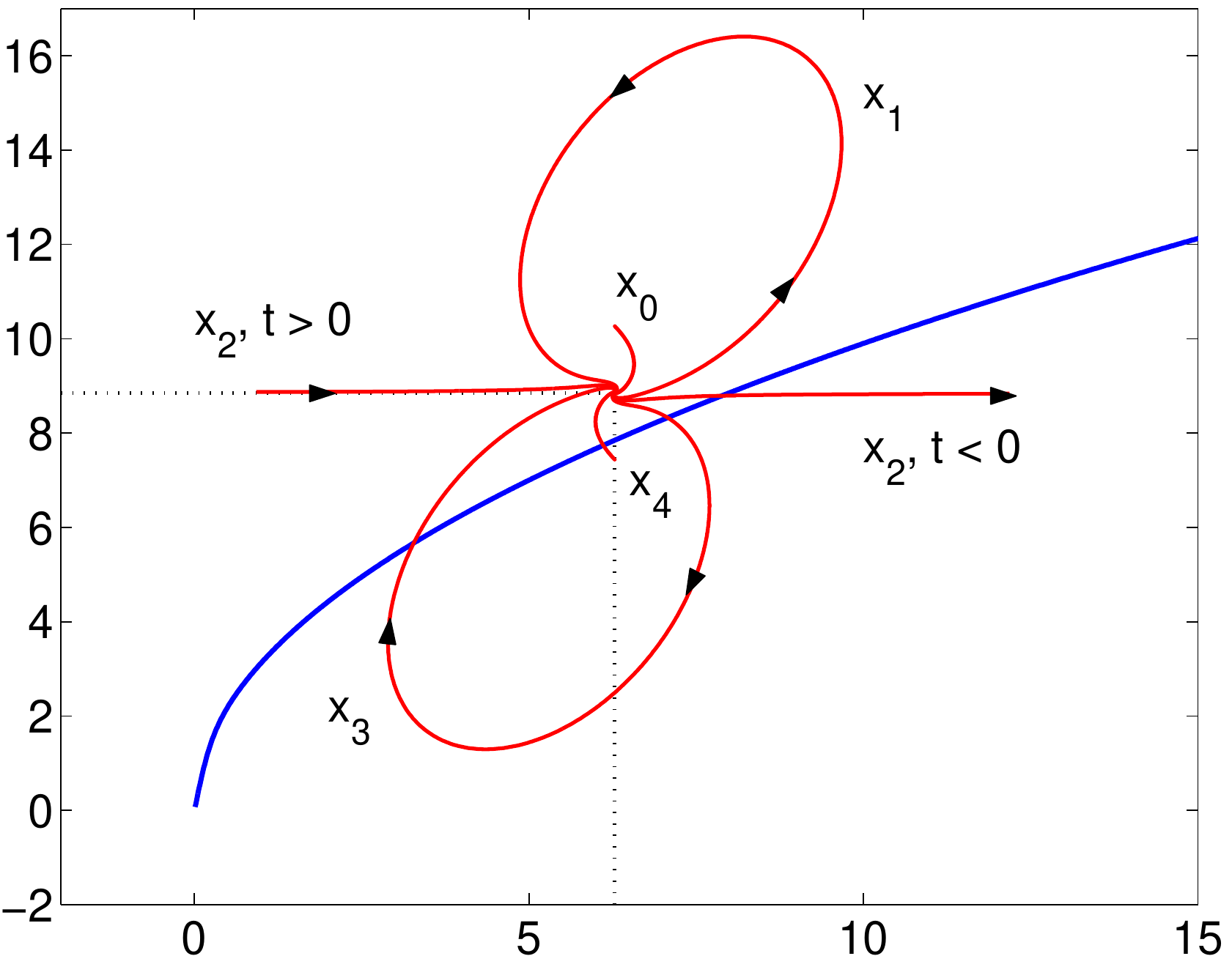}
\end{center}
\caption{Trajectories in the dispersion plane of local frequency vs. local wavenumber of the physical SFB. In the left picture, at various times the trajectories are shown parameterized by $x$; the \textit{phase singularity} is observed for $t_{0}=0$. In the right picture, the trajectories parameterized by~$t$ at various positions, with \textit{phase singularity} at $x_{2}=\pi /2$.} \label{trajectories}
\end{figure}

\section[Phase--Amplitude Equations for the NLS Equation]{Phase--Amplitude Equations for the NLS\\ Equation}

To find the \textit{phase--amplitude} (PA) equations for the NLS equation, we substitute $\psi(x,t) = a(x,t)\,e^{i\,\theta(x,t)}$ to equation (\ref{NLS}), where $a(x,t)$ and $\theta(x,t)$ being real-valued functions. We end up with the expression
\begin{eqnarray}
\left(\frac{\partial a}{\partial t} + V_{0}\,\frac{\partial a}{\partial x} - 2\,\beta\,\frac{\partial \theta}{\partial x}\,\frac{\partial a}{\partial x} - \beta\,\frac{\partial^{2} \theta}{\partial x^{2}}\,a \right) + &&\nonumber \\ 
i\,\left(\left[\frac{\partial \theta}{\partial t} + V_{0}\,\frac{\partial \theta}{\partial x} - \beta\,\left(\frac{\partial \theta}{\partial x} \right)^{2}\right]a + \beta\,\frac{\partial^{2} a}{\partial x^{2}} + \gamma\,|a|^{2}a \right) \!\!\! &=&  \!\!\!0.
\end{eqnarray}
This leads to PA equation:~\cite{Sule99}
\begin{equation}
\left\{ {\begin{array}{*{20}c}
P &:& \left[\partial_{t} \theta + V_{0}\,\partial_{x} \theta - \beta\,\left(\partial_{x} \theta \right)^{2}\right]a + \beta\,\partial_{x}^{2} a + \gamma\,|a|^{2}a = 0,  \\
A &:& \partial_{t} a + V_{0}\,\partial_{x} a - 2\,\beta\,\partial_{x} \theta\,\partial_{x} a - \beta\,a\,\partial_{x}^{2}\theta = 0. \end{array}} \right.
\end{equation}
From this formulation, we can see that there is a coupling between the phase and the amplitude.

Consider the phase equation and substitute $\partial_{t} \theta = \omega_{0} - \omega = \Omega(k_{0}) - \omega$, $\frac{\partial \theta}{\partial x} = k - k_{0}$, $V_{0} = \Omega^{'}\!(k_{0})$, and $\beta = -\frac{1}{2}\,\Omega^{''}\!(k_{0})$ into the equation, to obtain
\begin{equation}
\Omega(k_{0}) - \omega + \Omega^{'}\!(k_{0})\,(k - k_{0}) + \frac{1}{2}\,\Omega^{''}\!(k_{0})\,(k - k_{0})^{2} = - \left(\beta\,\frac{\partial_{x}^{2}a}{a} + \gamma\,|a|^{2} \right).
\end{equation}
Observe that the terms on the left-hand side are the Taylor expansion of the linear dispersion relation $\Omega(k)$ about $k_{0}$. Therefore, the phase equation can be written in the form:
\begin{equation}
\omega - \Omega(k) + \textmd{$\cal{O}$}(k^{3}) =\beta\,\frac{\partial_{x}^{2}a}{a} + \gamma\,|a|^{2}. 				\label{NDR}
\end{equation}
This relation is known as the \textit{nonlinear dispersion relation}; and the term $\beta\,\frac{\partial_{x}^{2}a}{a}$ in equation~(\ref{NDR}) is sometimes also called \textit{Fornberg--Whitham term} \cite{Infe90}.

The expression $\omega - \Omega(k)$ on the left-hand side of~(\ref{NDR}) can be interpreted as the `dispersion-mismatch': the amount in which the local wavenumber and frequency do not satisfy the linear dispersion relation. In the previous dispersion plots, this corresponds to points that do not lie on the graph of the linear dispersion relation. For the `Fornberg--Witham term', vanishing real amplitude can cause an infinite difference between local wavenumber $\omega$ and the linear dispersion relation $\Omega(k)$, except for pure sinusoidal real amplitude. Therefore, in the dispersion plane, this corresponds to points far away from the linear dispersion relation. In Figure~\ref{trajectories} (left), the vanishing real amplitude happens at $t_{0} = 0$ and $x = \pi/2$, the local wavenumber becomes $\pm \infty$. For vanishing the `Fornberg--Whitham term', $\beta\,\frac{\partial_{x}^{2}a}{a} = 0$, this corresponds to a point above the linear dispersion relation $(k_{0},\Omega(k_{0})+\gamma\,|a|^2)$. Compare this ordinate value with the special case of the `nonlinear dispersion relation' on page~\pageref{7}, formula~(\ref{solutiondependontime}). This happens when $t \rightarrow \pm \infty$, or equals to the unperturbed wave train, the plane--wave solution of the NLS equation. For this simplest type of solutions of NLS, with plane wavefronts, the dispersion mismatch is constant, leading to a standard nonlinear pendulum equation for the amplitude, leading to the steady-state solutions like the soliton and nonlinear periodic oscillations near the equilibrium value.

Now consider the amplitude equation
\begin{equation}
\frac{\partial a}{\partial t} + \left[V_{0} - 2\,\beta\,\frac{\partial \theta}{\partial x}\right] \,\frac{\partial a}{\partial x} - \beta\,\frac{\partial^2 \theta}{\partial x^2}\,a = 0.
\end{equation}
Multiplying the equation with $a$, substitute $V_{0} = \Omega'(k_{0})$, $\beta = -\frac{1}{2}\Omega''(k_{0})$, and $\partial_{x}\theta = k - k_{0}$, we have the following expression
\begin{equation}
a\,\partial_{t}a + [\Omega'(k_{0}) + \Omega''(k_{0})\,(k - k_{0})]\,a\,\partial_{x}a + \frac{\partial}{\partial x}\left[\Omega'(k_{0}) + \Omega''(k_{0})\,(k - k_{0}) \right]\,\frac{1}{2}a^2 = 0.
\end{equation}
By recognizing that the second and third terms have linear Taylor expansion expression $\Omega'(k_{0})$ about $k_{0}$, we end up with the formulation
\begin{equation}
\partial_{t}\left(a^{2}\right) + \partial_{x}\left[\Omega'(k)\,a^{2}\right] = 0.
\end{equation}
This expression represents the \textit{energy equation}, with `energy density' $E = a^2$.

For the physical wave profile of SFB, the local wavenumber and local frequency are not constant (see Appendix~\ref{A}) and in fact, produce  large variations in the dispersion mismatch. In the linear systems, or for small amplitude waves in the nonlinear model, this corresponds to the almost vanishing of the denominator of the first term in the right-hand side in~(\ref{NDR}), i.e., $\beta \partial _{x}^{2}a/a$, when the profile differs from a pure sinusoidal and shifts the phase lines near wave dislocations. This is clearly seen in the plot of Figure~\ref{sfb_slices} which shows at various instants of time near the critical time $t = 0$ the real wave profile and the amplitude (dashed). Notice that when $t = 0$, the amplitude vanishes at $x = \pi/2$, leading to the singular behavior observed in the dispersion plots above also (see Figure~\ref{trajectories}). Therefore, we have the phase singularity phenomenon. In the nonlinear case, the focusing effect in between wave dislocations is enforced by the term $\gamma E$ in the phase equation, leading to the large amplitude amplification.	\newline
\begin{figure}[h]
\begin{center}
\includegraphics[width = 0.6\textwidth]{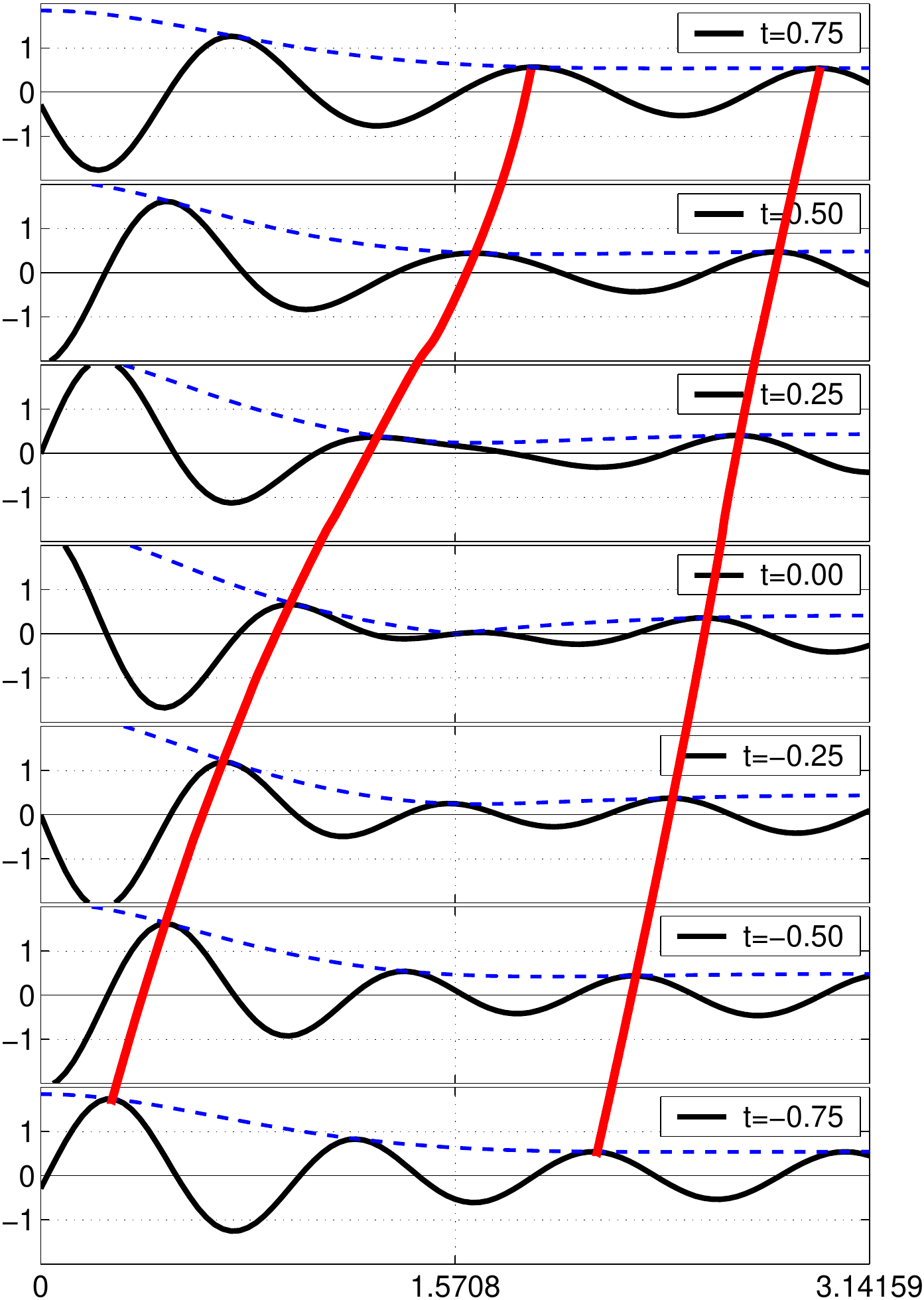}
\end{center}
\caption{Slices of the solution depicted in Figure~\ref{physicalSFB} showing detailed plots of the wave elevation (solid lines) and amplitude (dotted lines) as functions of position at various times near the wave dislocation at $(x,t) = (\pi/2,0)$. Observe the vanishing of the envelope at $(\pi/2,0)$ and the flattening of the surface in the focusing region, resulting in the disappearance of one wave when crossing $t = 0$, as shown by the curve connecting two wave crests.} 				\label{sfb_slices}
\end{figure}

\chapter{Extreme Waves in Perspective} 			\label{perspective}

\section{Bi-harmonic Wave Group Interaction}

In the previous two chapters, we have seen phenomena of wave group evolution that lead to extreme waves. One simple example of a wave group is bi-harmonic\footnote[1]{In literature, the term \textit{bichromatic} is also used to refer to the \textit{bi-harmonic} in this thesis. Normally, the term \textit{bichromatic} is used in light or optics to describe light with two different colors, from the Greek word \textit{chr$\bar{o}$ma} which means color.} wave group. This wave group is obtained by a linear superposition of two harmonic waves. The superposition of the same amplitude but different wavenumber or frequency produces a \textit{beat pattern}.
\begin{figure}[h]
\centering
\includegraphics[width = 15cm, height = 3.8cm, angle = 0]{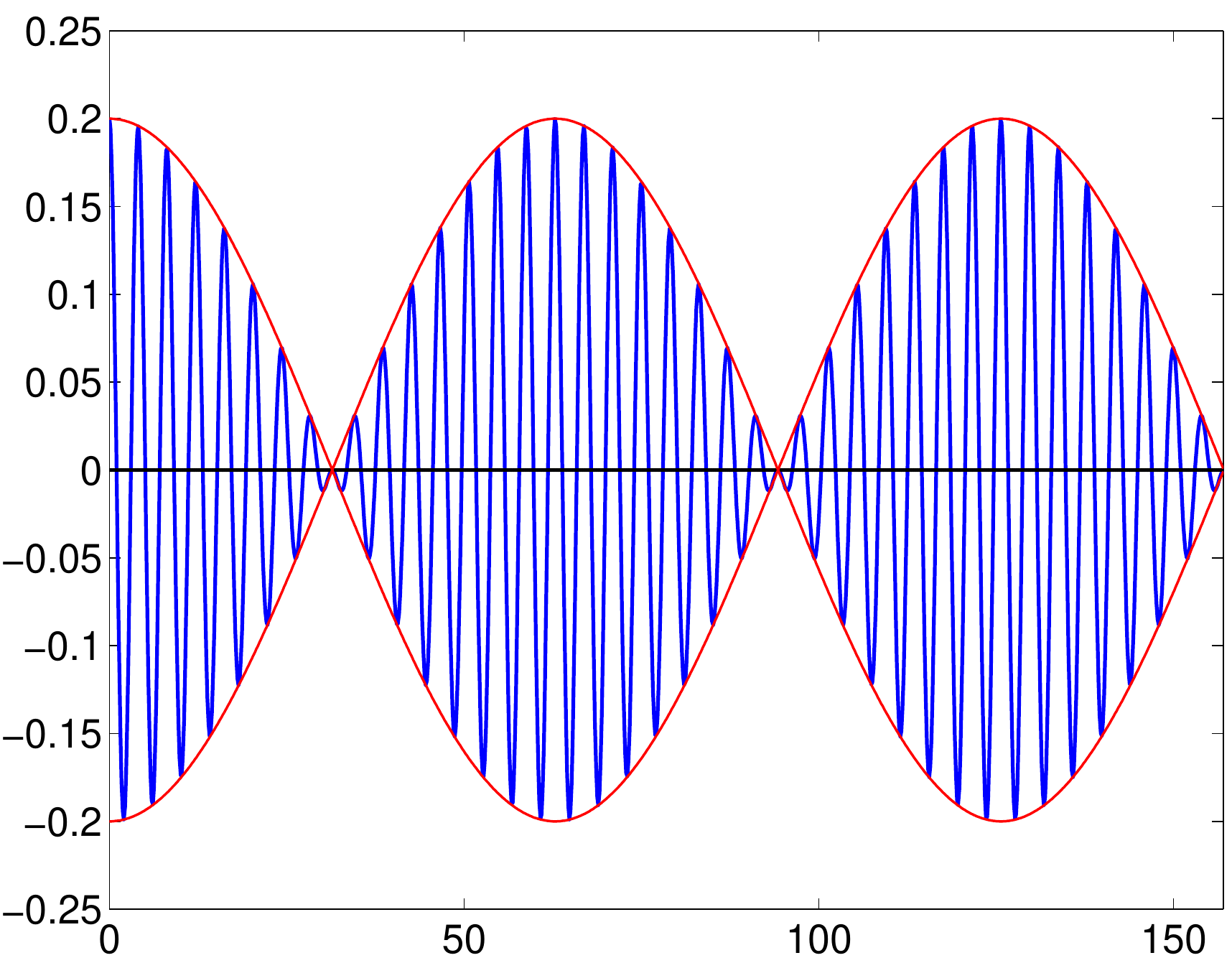}
\caption{A plot of the profile of a bi-harmonic wave group shows a \textit{beat pattern} as a result of the linear superposition of two harmonic waves. The blue profile is the carrier wave and the red profile is the envelope.}
\end{figure}

Now in this section, we will give a brief overview of wave group interaction. There are many references discussing soliton interaction, but specifically, we want to learn about the interaction of two bi-harmonic waves. For this purpose, we take the special case of two bi-harmonics with the same carrier wavelength. One example is given in Figure~\ref{eta0_car_m}.
\begin{figure}[h]
\centering
\includegraphics[width = 15cm, height = 3.8cm, angle = 0]{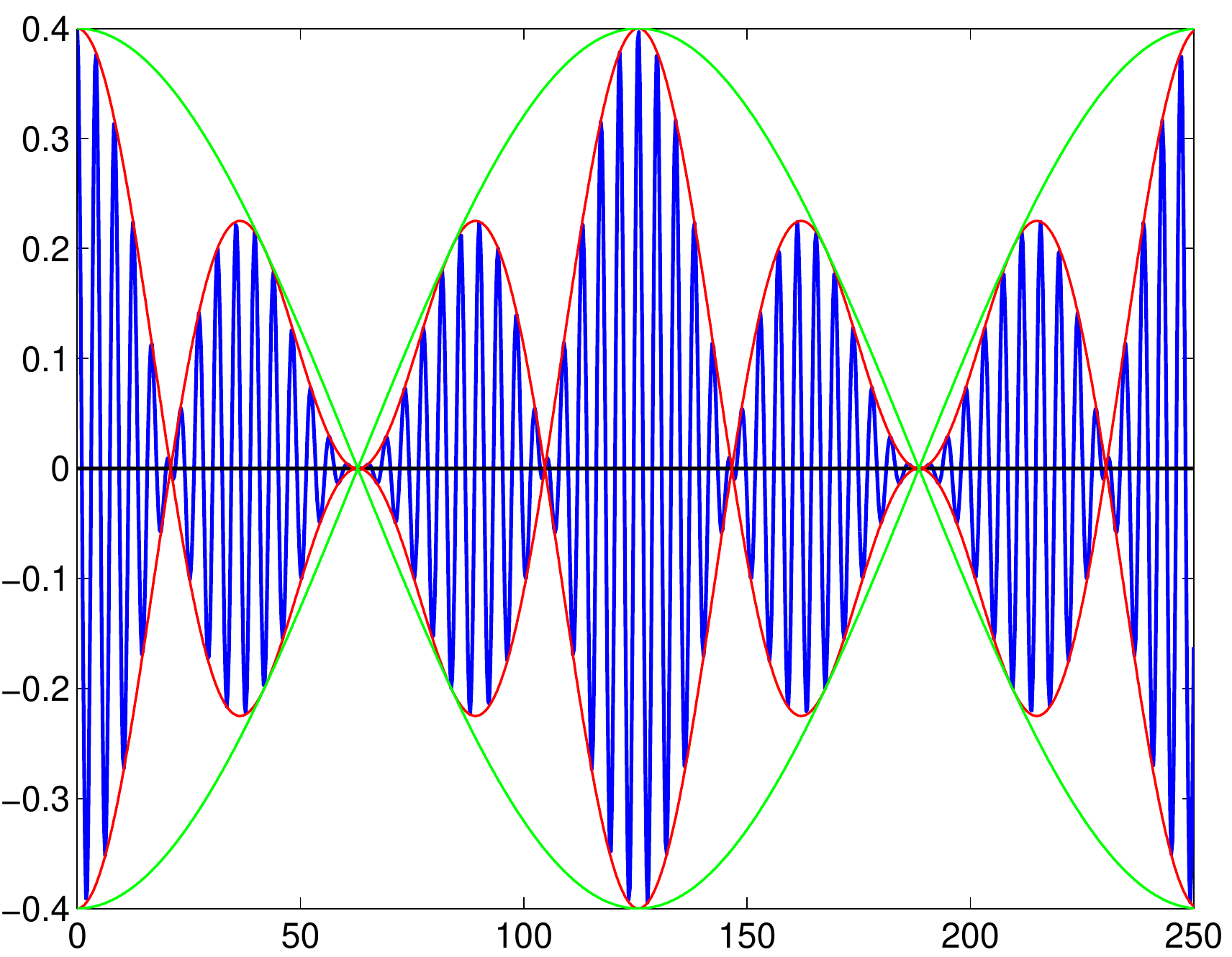}
\caption{A plot of the profile of two bi-harmonics with the same carrier wavelength when $t = 0$. The blue profile is the actual profile, the red profile is the envelope, and the green profile is the envelope of the envelope.} \label{eta0_car_m}
\end{figure}

This two bi-harmonic interaction is the linear superposition of four harmonic waves with four different wavenumbers and frequencies but with the special case that the two bi-harmonics have the same carrier wavelength. From the figure, we can see that apart from the envelope of the actual wave profile, we also have an envelope of the envelope. However, to determine the wavelength of the envelope or the envelope of the envelope is not arbitrary. It depends on the values of wavenumbers we choose. If the ratio of the envelope wavelength between the first and the second bi-harmonic is a rational number, then the envelope wavelength is the \textit{least common multiple} of both these envelope wavelengths.

\section[Bi-harmonic, Positively Modulated Bi-harmonic, and\\ Benjamin--Feir]{Bi-harmonic, Positively Modulated\\ Bi-harmonic, and Benjamin--Feir}

In this section, we want to compare the bi-harmonic (BH) wave group, positively modulated bi-harmonic (mod-BH) wave group, and Benjamin--Feir (BF) wave group. As already discussed in the previous section, the BH wave group can be obtained from the linear superposition of two sinusoidal waves with different wavenumbers or frequencies. However, we can write this wave group in two different representations, as an initial value problem, or as a signaling problem. In this section, we will use the latter. The signal of the BH wave group is given by
\begin{eqnarray}
\eta_{\textmd{\footnotesize BH}}(0,t) 
&=& q\,(\cos \omega_{1}t + \cos \omega_{2}t) \\
&=& \underbrace{2q\,\cos \mu t} \cdot \underbrace{\cos  \bar{\omega} t}, \label{BH_wavegroup} \\
& & \textmd{\footnotesize `envelope'} \quad \! \textmd{\footnotesize carrier} \nonumber
\end{eqnarray}
where $\bar{\omega} = \frac{1}{2}\left(\omega_{1} + \omega_{2}\right)$ and $\mu = \frac{1}{2}\left|\omega_{1} - \omega_{2}\right|$. From~(\ref{BH_wavegroup}), it is observed that the BH wave group can be interpreted as a carrier wave $\cos \bar{\omega} t$ modulated by an envelope wave $2q\,\cos \mu t$. The spectrum of the BH consists of two distinct frequencies $\omega_{1} = \bar{\omega} - \nu$ and $\omega_{2} = \bar{\omega} + \nu$, with the same amplitude $q$. Moreover, as we will see in the next paragraph for another type of wave group, the carrier wave of this BH wave group remains a smooth function and does not make a phase jump at the zero crossings of the so-called `envelope' BH. The evolution of the BH wave group as an initial signal of the `modified' KdV equation is discussed in detail by Cahyono~\cite{Cahy02}. The author uses a perturbation technique using series expansion. The thesis also motivates understanding of envelope deformations appearance and amplitude increase, as reported numerically in~\cite{West01}.

We know that the so-called `envelope' function of the BH wave group in~(\ref{BH_wavegroup}) is not always positive and therefore we must be careful to interpret this so-called `envelope' as an envelope of a wave group. We suggest to interpret the absolute value of the function as the envelope, and as a consequence, we have a non-negative envelope. It is referred to in~\cite{West01} that the BH wave group with a non-negative envelope as a positively modulated bi-harmonic wave group, or in short, mod-BH. The signal of mod-BH wave group is given by
\begin{eqnarray}
\eta_{\textmd{\footnotesize mod-BH}}(0,t) = 2q\,|\cos \mu t| \cdot \cos  \bar{\omega} t.
\end{eqnarray}
This signal is different from the BH wave group by a sign jump (or equivalently phase jump of $\pi$) of the carrier wave at the zero crossings of the so-called `envelope' BH. The numerical simulation has been studied by Westhuis~\cite{West01}, and there are some differences. Firstly, the mod-BH wave group propagates slightly faster than the BH wave group. Secondly, the wave heights due to the mod-BH wave group are a fraction higher than the BH wave group. Although there are some differences, the simulations also indicate that the evolution of the mod-BH and the BH are similar in the sense that both experience envelope deformation and amplitude increase. The spectrum of the mod-BH has most of its energy at the central frequency $\bar{\omega}$ and in sidebands $\bar{\omega} \pm \mu$.

Furthermore, another signal that also has a spectrum with three distinct frequencies is a regular wave modulated by two symmetrical sidebands. We call this signal as Benjamin--Feir (BF) wave group signal, with the central frequency $\bar{\omega}$ and side bands frequencies $\bar{\omega} \pm \nu$. In 1967, Benjamin and Feir showed that these modulations are unstable and will grow exponentially in time~\cite{Debn94}. Notice also that we take the BF wave group and study its evolution as our model for extreme wave generation (see Chapter~\ref{BF_SFB}). The signal can be formulated as follows:
\begin{equation}
\eta_{\textmd{\footnotesize BF}}(0,t) = q\,(1 + \epsilon\,\cos \nu t)\cdot\cos  \bar{\omega} t, 				\label{signal_BF}
\end{equation}
where $\epsilon$ is a small perturbation.
For the special (extreme) case of BF wave, when $\nu = 2\,\mu$ and $\epsilon = 1$,
the signal (\ref{signal_BF}) can be written as
\begin{eqnarray}
\eta_{\textmd{\footnotesize BF}}(0,t) &=&  q\,(1 + \cos 2\mu t)\cdot\cos  \bar{\omega} t \\
									  &=& 2q\,\cos\,^{2}\mu t \cdot \cos \bar{\omega} t.
\end{eqnarray}
From this expression, we can see that there exists a relation between the extreme case BF and the mod-BH since both of them have the same carrier wave but a different envelope. We know from Chapter~\ref{BF_SFB} that the nonlinear evolution of BF wave can be described by the Soliton on Finite Background of the NLS equation, and Chapter~\ref{WD_PS} shows that BF has phase singularity. According to the BF instability, the BH (and also mod-BH) will not grow if $(\mu/\bar{\omega})/(kq) > \sqrt{2}$, but Westhuis~\cite{West01} shows by numerical calculation that the growing sidebands is present for the BH. Therefore, the BF and mod-BH wave groups have different behavior in the instability.

\section{Extreme Waves in a Random Sea}

In Chapter~\ref{BF_SFB} and Chapter~\ref{WD_PS} we have studied the generation of extreme waves using the NLS equation, where the exact solution can be determined by taking a nonlinear extension of the BF instability. The exact analytic solutions are also found by Osborne et al.~\cite{Osbo00} using a different method, i.e. the \textit{inverse scattering transform} (IST) method. Two simple examples of an infinite class of more general unstable solutions of the NLS equation are given. One of these two simple solutions has spatial periodicity, similar to our SFB for maximum growth rate, with amplitude amplification factor~$\approx 2.4$. Another solution has temporal periodicity with amplitude amplification factor~$\approx 3.8$. The same authors also perform extensive numerical simulations for two dimensional NLS equation. The numerical results show that unstable modes do indeed exist in two dimensional NLS equation, and they can take the form of large amplitude extreme waves. The authors of that paper also provide evidence that, for a wide range of initial conditions, `coherent structures' in two dimensions (or `unstable modes' in one dimension) do indeed exist in the surface wave field and are ubiquitous in the simulations of deep water gravity waves. However, these structures are not strictly solitons, but are instead `unstable modes' or `extreme waves'.

However, several papers have been published to discuss the extreme wave generation using numerical simulations of the NLS equation with initial conditions typical examples of random oceanic sea states described by the Joint North Sea Project (JONSWAP) power spectrum \cite{Onor00b, Onor01, Onor02}. The spectrum is given by the formula
\begin{equation}
P(f) = \frac{\alpha}{f^5} \, \textmd{exp} \left[-\frac{5}{4} \left(\frac{f_{0}}{f} \right)^{4} \right] \, \gamma^{\textmd{exp} - [-[(f - f_{0})^2]/(2\sigma_{0}^{2} f_{0}^{2})]},
\end{equation}
where $\sigma_{0} = 0.07$ if $f \leq f_{0}$ and $\sigma_{0} = 0.09$ if $f > f_{0}$. This spectrum has a dominant frequency $f_{0}$, and two parameters: \textit{`enhancement' coefficient} $\gamma$ and \textit{Phillips parameter} $\alpha$. As $\gamma$ increases, the spectrum becomes higher and narrower around the spectral peak. In Figure~\ref{JONSWAP}, we show the JONSWAP spectrum for different values of $\gamma$ ($\gamma = 1, 5, 10$) for $f_{0} = 0.1$ Hz and $\alpha = 0.0081$.
\begin{figure}[h]
\begin{center}
\includegraphics[width = 0.6\textwidth]{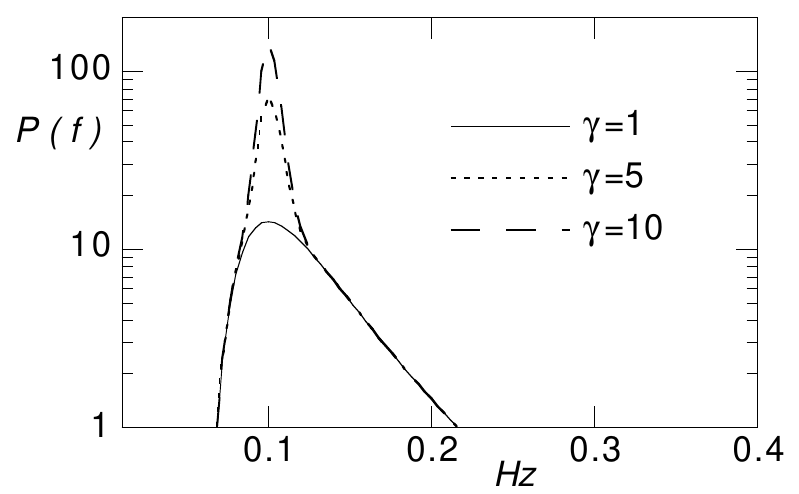} \vspace*{-0.5cm}
\end{center}
\caption{The JONSWAP spectrum for $\gamma = 1$ (solid line), $\gamma = 5$ (dotted line), and $\gamma = 10$ (dashed line); with $f_{0} = 0.1$ Hz and $\alpha = 0.0081$. (Courtesy of Onorato et al.~\cite{Onor01}.)} 				\label{JONSWAP}
\end{figure}

Onorato et al. in~\cite{Onor00b} use the NLS equation and its higher-order correction \textit{Dysthe--Lo--Mei} (DLM) equation as simple models for explaining the generation of extreme waves in one dimension. By performing a standard linear stability analysis of the NLS equation under the hypothesis of a small amplitude perturbation, the authors give a formula for amplitude amplification factor as a function of wave steepness and number of waves under the perturbation. By comparing this factor from the Inverse Scattering Theory, NLS equation, and DLM equation, it is found that as the steepness and number of waves under envelope are increasing, the amplitude amplification can be as large as a factor of three. See Figure~\ref{NLS_DLM_IST} for the plots of the amplitude amplification factor. However, when the amplitude of the perturbation is relatively large, Onorato et al. show numerically in~\cite{Onor00a} that the amplitude amplification can reach a factor~$\approx 3.5$.\\
\begin{figure}[h]
\begin{center}
\hspace*{-0.1cm}
\includegraphics[width = 0.45\textwidth]{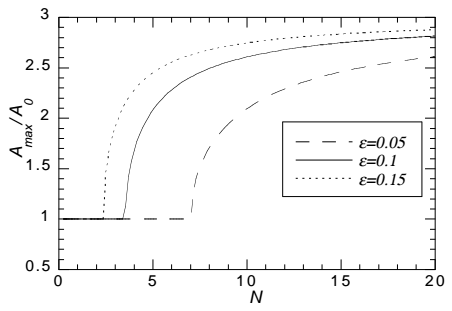}   		\hspace{0.5cm}
\includegraphics[width = 0.48\textwidth]{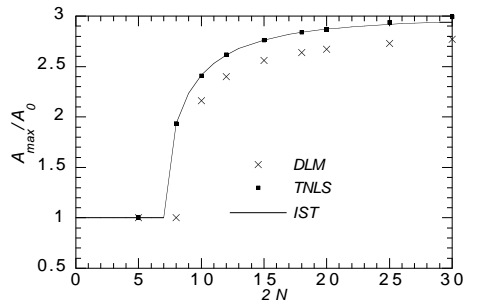}			\vspace*{-0.5cm}
\end{center}
\caption{The amplitude amplification factor as a function of the number of waves under the perturbation $N$ for different values of steepness $\epsilon = 0.05$, 0.1, and 0.15 (left). The amplitude amplification factor from the Inverse Scattering Theory (IST) and a numerical simulation of the NLS and DLM equations, with the value of initial steepness $\epsilon = 0.1$ (right). (Courtesy of Onorato et al.~\cite{Onor00b}.)}  					\label{NLS_DLM_IST}
\end{figure}

From numerical simulations using random-wave initial conditions characterized by the JONSWAP power spectrum, it is found that for increasing parameters in the spectrum (the `enhancement' coefficient and Phillips parameter), the effects of nonlinearity become more important, leading to the occurrence of extreme waves. The same result is also shown in~\cite {Onor01} numerically: extreme waves in a random sea state are more likely to occur for large values of the parameters in the spectrum.

Furthermore, Onorato et al.~\cite{Onor02} continue to investigate the appearance of extreme waves from the two-dimensional NLS equation, as considered in~\cite{Osbo00}. The authors study a more realistic case characterized by the JONSWAP spectrum for a directional wave train, while at the same time extending the order of the simulations. To investigate many of the effects of directional spreading on the generation of extreme waves, the authors choose to study higher-order formulation that shares many properties of the Zakharov equation, the so-called \textit{`generalized Dysthe equation'}.

\chapter{Conclusions} 						\label{conclusions}

In this thesis, we only examine one particular model to understand several aspects of the `Extreme Waves' project. We studied that under very long perturbation, the uniform wave train can linearly grow in time, leads into modulational instability phenomenon, in this case, the Benjamin--Feir (BF) instability. However, the nonlinear effect will bind this exponential growth, as described by an exact solution of the NLS equation, namely the Soliton on Finite Background (SFB).

We also showed that the maximum factor of amplitude amplification due to the BF instability can become as large as a factor of three. This amplitude amplification factor does not depend on the growth rate of the BF instability, but it grows monotonically for increasing the  perturbation wavelength. Furthermore, the SFB physical wave profile of the NLS equation shows the wave dislocation phenomenon for $\kappa < \sqrt{3/2}$. This phenomenon can be observed when at a specific position and specific time two waves are merging into one wave or one wave is splitting into two waves.

This phenomenon also related to another phenomenon, i.e., the phase singularity. The phase singularity happens when the local wavenumber and the local frequency become unbounded when the wave amplitude vanishes. The trajectories of the local wavenumber and the local frequency in the dispersion plane depicted the singular behavior. It is also found that the wave dislocation phenomenon does not happens for all perturbation wavenumber $\kappa$ but only for $0 < \kappa < \kappa_{\textmd{\footnotesize crit}}$, where $\kappa_{\textmd{\footnotesize crit}} = \sqrt{3/2}$ for the normalized SFB.

From the overview of extreme waves in random oceanic sea state, we learn that the JONSWAP power spectrum as an initial condition with large values of parameters will support the occurrence of extreme waves generation. Moreover, from several models that they take for extreme wave generation, it is obtained that the maximum amplitude amplification is also three although some numerical results can give a factor larger than three. Therefore, we conclude that no matter what appropriate models we choose for extreme wave generation, we will get a factor three for amplitude amplification.

\chapter{Open Problems and Future Research} 					\label{open_problems}

\begin{itemize}
\item We already know that for our NLS equation model, the maximum amplitude amplification factor can reach the value three. As shown in \cite{Onor00b}, even by taking some other models, the `Inverse Scattering' theory and the `Dysthe--Lo--Mei' (DLM) equation, the factor cannot exceed the number of three. We want to show by other models that the maximum amplitude amplification factor will also three.

\item To search the relationship between the amplitude amplification factor and the breaking phenomenon. In nature, some of the waves would break before they reach the maximum amplitude. Therefore, it is difficult to observe the maximum amplitude amplification factor in hydrodynamic laboratories.

\item The PhD thesis \cite{Cahy02} has studied the evolution of the bi-harmonic wave group based on the `modified' KdV equation. We want to study the evolution of two bi-harmonic wave groups based on this equation or other models.
\end{itemize}

\appendix

\chapter{Some Formulas Related to Soliton on Finite Background} \label{A}

The `Soliton on Finite Background' (SFB) in non-normalized form is given by:
\begin{eqnarray}
\!\!\!\!\!\!\!\!\widetilde{\psi}(\xi,\tau) \!\!\!&:=& \!\!r_{0}\left(\frac{(\kappa^{2} - 1)\, \cosh(\kappa\,\sqrt{2 - \kappa^2}\,\gamma\,r_{0}^2\,\tau) + \sqrt{1 - \frac{1}{2}\,\kappa^{2}}\, \cos\left(\kappa\,r_{0}\,\sqrt{\frac{\gamma}{\beta}}\, \xi\right)}{\cosh(\kappa\,\sqrt{2 - \kappa^2}\,\gamma\,r_{0}^2\, \tau) - \sqrt{1 - \frac{1}{2}\,\kappa^{2}} \cos\left(\kappa\,r_{0}\,\sqrt{\frac{\gamma}{\beta}}\, \xi\right)}\right. \nonumber
\\\!\!\!\! &-& \!\!\!\! i\,\left. \frac{\kappa\,\sqrt{2-\kappa^2}\, \sinh(\kappa\,\sqrt{2-\kappa^2}\,\gamma\,r_{0}^2\, \tau)} {\cosh(\kappa\,\sqrt{2-\kappa^2}\,\gamma\,r_{0}^2\, \tau) - \sqrt{1 - \frac{1}{2}\,\kappa^{2}}\, \cos\left(\kappa\,r_{0}\,\sqrt{\frac{\gamma}{\beta}}\, \xi\right)} \right)e^{\,-i\,\gamma\,r_{0}^2\,\tau}, \label{exact_lengkap}
\end{eqnarray}
where $\sigma(\kappa) = \kappa \sqrt{2 - \kappa^2}$ is the growth rate corresponding to the SFB in normalized form. Writing this SFB in the phase-amplitude form, $\widetilde{\psi}(\xi,\tau) = \widetilde{a}(\xi,\tau)\,e^{\,i\,\widetilde{\theta}(\xi,\tau)}$, then the real-valued amplitude reads
\begin{equation}
\widetilde{a}(\xi,\tau) =r_{0}\frac{\sqrt{C_{1}}}{C_{2}},
\end{equation}
where
\begin{eqnarray}
C_{1} &=& \left[(\kappa^2 - 1)\cosh(\gamma\,r_{0}^2\,\sigma(\kappa)\,\tau) + \sqrt{1 - \frac{1}{2} \kappa^2} \cos\left(\kappa \, r_{0}
\, \sqrt{\frac{\gamma}{\beta}} \, \xi \right) \right]^2 \nonumber\\
&+&\sigma^2(\kappa)\,\sinh^2(\gamma\,r_{0}^2\,\sigma(\kappa)\,\tau),
\end{eqnarray}
and
\begin{equation}
C_{2} = \cosh(\gamma\,r_{0}^2\,\sigma(\kappa)\,\tau) -\sqrt{1 - \frac{1}{2} \kappa^2} \, \cos\left(\kappa \, r_{0} \, \sqrt{\frac{\gamma}{\beta}} \, \xi\right).
\end{equation}
The real-valued phase reads
\begin{equation}
\widetilde{\theta}(\xi,\tau) = -\gamma\,r_{0}^2\,\tau + \textmd{tan}^{-1} \left(\frac{-\sigma(\kappa)\,\sinh(\gamma\,r_{0}^2\,\sigma(\kappa)\,\tau)} {\left(\kappa^2 - 1 \right) \, \cosh(\gamma\,r_{0}^2\,\sigma(\kappa)\,\tau) +\sqrt{1 - \frac{1}{2} \kappa^2} \, \cos\left(\kappa \, r_{0} \, \sqrt{\frac{\gamma}{\beta}} \, \xi \right)} \right).
\end{equation}
From the definition of the local wavenumber~(\ref{local_wavenumber}) and local frequency~(\ref{local_frequency}), we can also derive the expression for these.\\
The local wavenumber reads
\begin{eqnarray}
&\!\!\!\!\!\!\!&\!\!\!\!\!\!\!\!\!\!\!\!\!\!k(x,t) = k_{0}  \nonumber \\
&-&\frac{r_{0}\,\sqrt{\frac{\gamma}{2\,\beta}}\,\sigma^2(\kappa)\, \sin\left(\kappa\,r_{0}\,\sqrt{\frac{\gamma}{\beta}}\,\xi\right) \,\sinh(\gamma\,r_{0}^2\,\sigma(\kappa)\,\tau)} {\left[\left(1-\frac{\sigma^2(\kappa)}{\kappa^2}\right)\,\cosh(\gamma\,r_{0}^2\,\sigma(\kappa)\,\tau) +\frac{\sigma(\kappa)}{\kappa\,\sqrt{2}}\,\cos\left(\kappa\,r_{0}\,\sqrt{\frac{\gamma}{\beta}}\,\xi\right)\right]^2 +\sigma^2(\kappa)\,\sinh^2(\gamma\,r_{0}^2\,\sigma(\kappa)\,\tau)}, \nonumber \\
\end{eqnarray}
and the local frequency reads
\begin{eqnarray}
&\!\!\!\!\!\!\!&\!\!\!\!\!\!\!\!\!\!\!\!\!\!\omega(x,t) = \omega_{0} + \gamma\,r_{0}^2 \nonumber \\
\!\!\!\!\!\!\!&+&\frac{\gamma\,r_{0}^2\,\sigma^2(\kappa)\,\cosh(\gamma\,r_{0}^2\,\sigma(\kappa)\,\tau) \left[\left(1-\frac{\sigma^2(\kappa)}{\kappa^2}\right)\,\cosh(\gamma\,r_{0}^2\,\sigma(\kappa)\,\tau) +\frac{\sigma(\kappa)}{\kappa\,\sqrt{2}}\,\cos\left(\kappa\,r_{0}\,\sqrt{\frac{\gamma}{\beta}}\,\xi\right)\right]} {\left[\left(1-\frac{\sigma^2(\kappa)}{\kappa^2}\right)\,\cosh(\gamma\,r_{0}^2\,\sigma(\kappa)\,\tau) +\frac{\sigma(\kappa)}{\kappa\,\sqrt{2}}\,\cos\left(\kappa\,r_{0}\,\sqrt{\frac{\gamma}{\beta}}\,\xi\right)\right]^2
+ \sigma^2(\kappa)\,\sinh^2(\gamma\,r_{0}^2\,\sigma(\kappa)\,\tau)} \nonumber \\
\!\!\!\!\!\!\!&-&\frac{\gamma\,r_{0}^2\,\sigma^2(\kappa)\,\left(1-\frac{\sigma^2(\kappa)}{\kappa^2}\right)\, \sinh^2(\gamma\,r_{0}^2\,\sigma(\kappa)\,\tau)}{\left[\left(1-\frac{\sigma^2(\kappa)}{\kappa^2}\right)\, \cosh(\gamma\,r_{0}^2\,\sigma(\kappa)\,\tau)+\frac{\sigma(\kappa)}{\kappa\,\sqrt{2}}\,\cos\left(\kappa\,r_{0}\, \sqrt{\frac{\gamma}{\beta}}\,\xi\right)\right]^2+\sigma^2(\kappa)\,\sinh^2(\gamma\,r_{0}^2\,\sigma(\kappa)\,\tau)}, \nonumber \\
\end{eqnarray}
with $\xi = x - \Omega'(k_{0})t$ and $\tau = t$.

\renewcommand{\bibname} {References}


\begin{thebibliography}{99}

\bibitem{Akhm97} Akhmediev, A. A., Ankiewicz, A. \textit{Solitons---Nonlinear Pulses and Beams}, volume 5 of \textit{Optical and Quantum Electronic Series}. Chapman \& Hall, first edition, 1997.

\bibitem{Cahy02} Cahyono, E. \textit{Analytical Wave Codes for Predicting Surface Waves in a Laboratory Basin}. PhD thesis, University of Twente, Department of Applied Mathematics, June 2002.  \addcontentsline{toc}{chapter}{References}

\bibitem{Dean90} Dean, R. G., \textit{Water Wave Kinematics}. Kluwer, Amsterdam, 1990, pp. 609--612.

\bibitem{Debn94} Debnath, L. \textit{Nonlinear Water Waves}. Academic Press, Inc., 1994.

\bibitem{Groe98} van Groesen, E. W. C. Wave groups in uni--directional surface--wave model. \textit{Journal of Engineering Mathematics}, 34(1-2):215-226, 1998.

\bibitem{Groe02} van Groesen, E. W. C., Karjanto, N., Peterson, P., Andonowati. Wave dislocation and nonlinear amplitude amplification for extreme fluid surface waves. Submitted to \textit{Physics Letter A}, December 2002.

\bibitem{Infe90} Infeld, E., Rowlands, G. \textit{Nonlinear waves, solitons, and chaos}, Cambridge University Press, 1990.

\bibitem{Karj02} Karjanto, N., van Groesen, E. W. C., Peterson, P. Investigation of the maximum amplitude increase from the Benjamin--Feir instability. \textit{Journal of Indonesian Mathematical Society (Majalah Ilmiah Himpunan Matematika Indonesia, MIHMI)}, 8(4):39-47, 2002.

\bibitem{KipD00} Kip, D., Soljacic, M., Segev, M., Eugenieva, E., Christodoulides, D. N. Modulation Instability and Pattern Formation in Spatially Incoherent Light Beams. \textit{Science} 290:495-498, 2000.

\bibitem{Onor00a} Onorato, M., Osborne, A. R., Serio, M. Nonlinear Dynamics of Rogue Waves. \textit{Sixth International Workshop  on Wave Hindcasting and Forecasting}, pp. 470-479, November 2000.

\bibitem{Onor00b} Onorato, M., Osborne, A. R., Serio, M., Daminani, T. Occurrence of Freak Waves from Envelope Equations in Random Ocean Wave Simulations. \textit{Proceedings of Rogue Waves 2000}, edited by M. Olagnon and G. A. Athanassoulis, Ifremer, Brest, France, 29-30 November 2000.

\bibitem{Onor01} Onorato, M., Osborne, A. R., Serio, M., Bertone, S. Freak Waves in Random Oceanic Sea States. \textit{Physical Review Letters}, 86(25):5831-5834, June 2001.

\bibitem{Onor02} Onorato, M., Osborne, A. R., Serio, M. Extreme waves events in directional, random oceanic sea states. \textit{Physics of Fluid}, 14(4):L25-L28, April 2002.

\bibitem{Osbo00} Osborne, A. R., Onorato, M., Serio, M. The nonlinear dynamics of rogue waves and holes in deep-water gravity wave trains. \textit{Physics Letters A}, 275:386-393, October 2000.

\bibitem{Osbo01} Osborne, A. R. The random and deterministic dynamics of `rogue waves' in unidirectional, deep-water wave trains. \textit{Marine Structures}, 14:275-293, 2001.

\bibitem{Scot99} Scott, A., \textit{Nonlinear Science---Emergence and Dynamics of Coherent Structures}. Oxford University Press, 1999.

\bibitem{Sule99} Sulem, C., Sulem, P. L. \textit{The Nonlinear Schr\"{o}dinger Equation --- Self--Focusing and Wave Collapse}, volume 139 of \textit{Applied Mathematical Sciences}. Springer, 1999.

\bibitem{West01} Westhuis, J.H. \textit{The Numerical Simulation of Nonlinear Waves in Hydrodynamic Model Test Basin}. PhD thesis, University of Twente, Department of Applied Mathematics, May 2001.

\bibitem{freakwave} British Broadcasting Corporation (BBC)--Science and Nature--Horizon's homepage, Freak Wave--Programme Summary: \url{http://www.bbc.co.uk/science/horizon/2002/freakwave.shtml}. Last accessed 13 June 2003.

\bibitem{freakwaveqa} BBC--Science and Nature--Horizon's homepage, Freak Wave--Questions and Answers: \url{http://www.bbc.co.uk/science/horizon/2002/freakwaveqa.shtml}. Last accessed 13 June 2003.
\end{thebibliography}
\end{document}